\PassOptionsToPackage{hypertexnames=false}{hyperref}
\documentclass{theoretics}

\usepackage{bbm}
\usepackage{cancel}
\usepackage{float}
\usepackage{framed}
\usepackage{manfnt}
\usepackage{mathrsfs}
\usepackage{mathtools}
\usepackage{stmaryrd}
\usepackage{verbatim}
\usepackage{xparse}
\usepackage{xspace}

\usepackage{thm-restate} 

\usepackage[nameinlink]{cleveref}
\crefname{ineq}{inequality\@}{inequalities\@}
\creflabelformat{ineq}{#2{\upshape(#1)}#3}

\let\abs\relax

\let\norm\relax

\let\Pr\relax

\let\set\relax

\DeclareMathOperator{\BC}{BC}

\DeclareMathOperator{\Fid}{F}

\DeclareMathOperator{\by}{\times}

\DeclareMathOperator{\poly}{poly}
\DeclareMathOperator{\polylog}{polylog}

\DeclareMathOperator{\tr}{tr}

\DeclareMathOperator{\unit}{\mathbbm{1}}

\DeclareMathOperator*{\E}{\mathbf{E}}
\DeclareMathOperator*{\Pr}{\mathbf{P}}
\DeclareMathOperator*{\Var}{\mathbf{Var}}
\DeclareMathOperator*{\stddev}{\mathbf{stddev}}

\DeclareMathOperator*{\argmin}{argmin}

\DeclareMathOperator{\Binomial}{Binomial}
\DeclareMathOperator{\Exponential}{Exponential}

\DeclarePairedDelimiter{\abs}{\lvert}{\rvert}

\DeclarePairedDelimiter{\bracket}{\lbrack}{\rbrack}

\DeclarePairedDelimiter{\norm}{\|}{\|}
\DeclarePairedDelimiter{\paren}{\lparen}{\rparen}
\DeclarePairedDelimiter{\set}{\{}{\}}

\ThCSnewtheoita{notation}

\Crefname{claim}{Claim}{Claims} 

\newcommand*{\eps}{\epsilon}

\newcommand*{\ol}{\overline}

\newcommand*{\tensor}{\otimes}

\newcommand*{\wh}[1]{\widehat{#1}}
\newcommand*{\wt}[1]{\widetilde{#1}}

\renewcommand*{\subset}{\subseteq}

\DeclareDocumentCommand{\Bin}{o}{
  \IfNoValueTF{#1}{
    \set{0, 1}
  }{
    \set{0, 1}^{#1}
  }
}

\DeclareDocumentCommand{\half}{o}{
  \IfNoValueTF{#1}{
    \frac{1}{2}
  }{
    \frac{#1}{2}
  }
}

\DeclareDocumentCommand{\dist}{o}{%
  \IfNoValueTF{#1}{d}{d_{\mathrm{#1}}}%
}

\DeclareDocumentCommand{\pd}{o m O{}}{
  \IfNoValueTF{#1}{
    \frac{\partial^{#3}}{\partial {#2}}
  }{
    \frac{\partial^{#3} {#1}}{\partial {#2}}
  }
}

\DeclareDocumentCommand{\td}{o m O{}}{
  \IfNoValueTF{#1}{
    \frac{d^{#3}}{d {#2}}
  }{
    \frac{d^{#3} {#1}}{d {#2}}
  }
}

\DeclareDocumentCommand{\note}{o m}{
  \IfNoValueTF{#1}{
    \marginpar{\tiny \textsf{#2}}
  }{
    \marginpar{\tiny \textsf{\textbf{#1}: #2}}
  }
}

\DeclareDocumentCommand{\bs}{m}{
  {\boldsymbol #1}
}

\DeclareDocumentCommand{\restrict}{s m}{%
  \IfBooleanTF{#1}{%
    {\left.\kern-\nulldelimiterspace{#2}\vphantom{\big|}\right|}%
  }{%
    {\left.\kern-\nulldelimiterspace{#2}\right|}%
  }%
}

\DeclareDocumentCommand{\ketbra}{m m}{%
  \lvert #1 \rangle \langle #2 \rvert%
}

\NewDocumentCommand{\Prob}{e{_} m}{%
  \IfNoValueTF{#1}{%
    \Pr \set*{#2}
  }{%
    \Pr_{#1} \set*{#2}
  }
}

\NewDocumentCommand{\Exp}{e{_} m}{%
  \IfNoValueTF{#1}{%
    \E \set*{#2}
  }{%
    \E_{#1} \set*{#2}
  }
}

\newcommand*{\CC}{\mathbb{C}}

\newcommand*{\NN}{\mathbb{N}}

\newcommand*{\RR}{\mathbb{R}}

\newcommand*{\cB}{\mathcal{B}}

\newcommand*{\cM}{\mathcal{M}}

\newcommand*{\cT}{\mathcal{T}}

\newcommand{\rnote}[1]{}
\newcommand{\cnote}[1]{}

\newcommand{\BajanSearch}{Threshold Search\xspace}
\newcommand{\nts}{n_{\mathrm{TS}}}

\usepackage{bbm}
\newcommand{\Id}{{\mathbbm 1}}

\title{Improved quantum data analysis}

\ThCSauthor{Costin B\u{a}descu}{cbadescu@cs.cmu.edu}[https://orcid.org/0009-0005-9389-0413]
\ThCSauthor{Ryan O'Donnell}{odonnell@cs.cmu.edu}[https://orcid.org/0000-0001-7608-1458]
\ThCSaffil{Computer Science Department, Carnegie Mellon University}

\ThCSthanks{ Supported by NSF grant FET-1909310 and ARO grant W911NF2110001.
    This material is based upon work
    supported by the National Science Foundation under grant numbers
    listed above. Any opinions, findings and conclusions or
    recommendations expressed in this material are those of the author
    and do not necessarily reflect the views of the National Science
    Foundation (NSF).}

  \ThCSshortnames{C. B\u{a}descu, R. O'Donnell}
\ThCSshorttitle{Improved quantum data analysis}

\ThCSyear{2024}
\ThCSarticlenum{7}
\ThCSreceived{Feb 7, 2023}
\ThCSrevised{Dec 12, 2023}
\ThCSaccepted{Jan 28, 2024}
\ThCSpublished{Mar 18, 2024}
\ThCSdoicreatedtrue

\ThCSkeywords{shadow tomography, adaptive data analysis, hypothesis testing}
\begin{document}
\maketitle

\begin{abstract}
   We provide more sample-efficient versions of some basic routines in quantum data analysis, along with simpler proofs.
   Particularly, we give a quantum ``Threshold Search'' algorithm that requires only $O((\log^2 m)/\eps^2)$ samples of a $d$-dimensional state~$\rho$.
   That is, given observables $0 \leq A_1, A_2, \dots, A_m \leq \Id$ such that $\tr(\rho A_i) \geq 1/2$ for at least one~$i$, the algorithm finds~$j$ with $\tr(\rho A_j) \geq 1/2-\eps$.
   As a consequence, we obtain a Shadow Tomography algorithm requiring only $\wt{O}((\log^2 m)(\log d)/\eps^4)$ samples, which simultaneously achieves the best known dependence on each parameter $m, d, \eps$.
   This yields the same sample complexity for quantum Hypothesis Selection among~$m$ states; we also give an alternative Hypothesis Selection method using $\wt{O}((\log^3 m)/\eps^2)$ samples.
\end{abstract}

\section{Introduction}  \label{sec:intro}
Some of the most basic problems in statistics, unsupervised learning, and property testing involve the following scenario:
One can observe data that are assumed to be drawn independently from an unknown probability distribution~$p$; say that $p$ is discrete and supported on $[d] = \{1, 2, \dots, d\}$.
The task is to learn, test, or estimate some properties of~$p$.
Completely estimating~$p$ up to error~$\eps$ (in, say, total variation distance) requires $\Theta(d/\eps^2)$ samples, so when~$d$ is very large one may seek to only learn or test \emph{partial} aspects of~$p$.
For example, one might only want to estimate the means of some known, fixed random variables $a_1, \dots, a_m : [d] \to [0,1]$ (sometimes called ``statistical queries'' in the learning/privacy literature).
Or, one might want to perform Hypothesis Selection over some set of two or more hypothesis distributions $q_1, \dots, q_m$ on~$[d]$.
It is generally fairly straightforward to determine the optimal sample complexity needed for these tasks.
For example, it's easy to show that one can simultaneously estimate all expectations $\E_p[a_1], \dots, \E_p[a_m]$ to accuracy~$\pm \eps$ using a batch of $n = O((\log m)/\eps^2)$ samples (independent of~$d$): one simply computes the empirical mean for each~$a_i$, reusing the batch of samples in each computation.

These kinds of questions become much more difficult to analyze when the classical source of randomness~$p$ is replaced by a \emph{quantum} source of randomness, namely a $d$-dimensional quantum state $\rho \in \CC^{d \times d}$ (satisfying $\rho \geq 0$, $\tr(\rho) = 1$).
The difficulties here are that: (i)~one cannot directly observe ``outcomes'' for~$\rho$, one can only measure it; (ii)~measuring the state~$\rho$ inherently alters it, hence reusing samples (i.e., copies of~$\rho$) is problematic.
For example, suppose we now have some known, fixed observables $A_1, \dots, A_m \in \CC^{d \times d}$ with $0 \leq A_i \leq \Id$ and we wish to estimate each expectation $\E_{\rho}[A_i] \coloneqq \tr(\rho A_i)$ to within~$\pm \eps$.
This is the ``Shadow Tomography'' problem introduced by Aaronson in~\cite{Aar16} (see~\cite{Aar20} for applications to, e.g., quantum money).
We do not know if this is similarly possible using $n = O((\log m)/\eps^2)$ copies of~$\rho$; indeed, prior to this work the best known upper bound was
\[
    n = \min\left\{\wt{O}((\log^4 m)(\log d)/\eps^4),\quad\wt{O}((\log^2 m)(\log^2 d)/\eps^8)\right\}.
\]
Here the sample complexity on the left is from~\cite{Aar20}, combining a ``Gentle Search'' routine with an online learning algorithm for quantum states from~\cite{ACHKN19}.
The sample complexity on the right was obtained by Aaronson and Rothblum~\cite{AR19} by drawing inspiration and techniques from the field of Differential Privacy.\footnote{See also~\cite{HKP20} for sample complexity bounds that can be better for special kinds of~$A_i$'s.}

In fact, we propose that --- rather than Differential Privacy --- a closer classical match for the Shadow Tomography problem is the task known as \emph{Adaptive Data Analysis}, introduced by~\cite{DFHPRR16}.
In this problem, the random variables (``statistical queries'') $a_1, \dots, a_m$ are not fixed in advance for the learner, but are rather received one at a time, with the crucial feature that each~$a_t$ may adaptively depend on the preceding estimates of $\E_p[a_1], \dots, \E_p[a_{t-1}]$ output by the learner.
In this case, conditioning on these output estimates skews the underlying i.i.d.\ product distribution~$p^{\otimes n}$ --- reminiscent of the way measuring a quantum state affects it --- and this prevents naive reuse of the sample data.
Indeed, it's far from obvious that the Adaptive Data Analysis task is doable with $\poly(\log m, \log d, 1/\eps)$ samples; however this was shown by~\cite{DFHPRR16}, who achieved complexity $n = \wt{O}((\log m)^{3/2} (\log d)^{1/2} / \eps^{7/2})$, and this was later improved by~\cite{BNSSSU15} to  $n = \wt{O}((\log m) (\log d)^{1/2} / \eps^{3}))$.
While Differential Privacy tools have been an ingredient in some Adaptive Data Analysis routines, the topics are not inherently linked; e.g., a viewpoint based on ``KL-stability'' is emphasized in~\cite{BNSSSU15}.

\subsection{Our work}

\subsubsection{\BajanSearch} \label{sec:threshold-search}
The first main result in our work concerns what we will call the quantum ``\BajanSearch'' problem.\footnote{Originally called the ``Secret Acceptor'' problem when it was introduced by Aaronson~\cite{Aar16}.  Later he called it ``Gentle Search''~\cite{Aar20}, but we find this name unsatisfactory as it is not necessary that a successful algorithm be ``gentle''.  In the Differential Privacy literature, it is sometimes called ``Report Noisy Max'' (offline case) or ``Above Threshold'' (online case)~\cite{DR13}.}
We state the problem here in a general form (recalling our notation $\E_\rho[A_i] =\tr(\rho A_i)$):

\paragraph{Quantum \BajanSearch problem:} \begin{samepage}\emph{Given:
\begin{enumerate}
\item Parameters $0 < \eps, \delta < \frac12$.
\item Access to unentangled copies of an unknown $d$-dimensional quantum state~$\rho$.
\item A list of $d$-dimensional observables $0 \leq A_1, \dots, A_m \leq \Id$.
\item A list of thresholds $0 \leq \theta_1, \dots, \theta_m \leq 1$.
\end{enumerate}
The algorithm should either output:
\begin{itemize}
    \item ``\,$\E_\rho[A_j] > \theta_j - \eps$'' for some particular $j$; or else,
     \item ``\,$\E_\rho[A_i] \leq \theta_i$ for all~$i$''.
\end{itemize}
}\end{samepage}
The output of the algorithm is a sample from a
distribution over indices $j$ such that
``\,$\E_\rho[A_j] > \theta_j - \eps$'' or
``\,$\E_\rho[A_i] \leq \theta_i$ for all~$i$'' if no such $j$
exists. The task is to minimize the number~$n$ of copies that are used,
while ensuring the probability of a false output statement is at
most~$\delta$.\\

We remark that all the difficulty of the problem is contained in the case where $\eps = \delta = \frac14$ and $\theta_j = \frac34$ for all~$j$ (see \Cref{sec:reductions}).
In this case, Aaronson~\cite{Aar16} originally showed that the \BajanSearch problem can be solved using~$n = \wt{O}(\log^4 m)$ copies of~$\rho$.
In the present paper, we improve this result quadratically:
\begin{theorem}                                     \label{thm:our-gentle-search}
    The quantum \BajanSearch problem can be solved using 
    \[
    \phantom{\qquad \text{($\textsc{l} = \log(1/\delta)$)}}
        n = \nts(m,\eps,\delta) = \frac{\log^2 m + \textsc{l}}{\eps^2} \cdot O(\textsc{l}) \qquad \textnormal{($\textsc{l} = \log(1/\delta)$)}
    \]
    copies of~$\rho$.
    Furthermore, this solution is \emph{online} in the sense that:
    \begin{itemize}
        \item The algorithm is initially given only $m, \eps, \delta$.  It then selects~$n$ and obtains $\rho^{\otimes n}$.
        \item Next, observable/threshold pairs $(A_1, \theta_1), (A_2, \theta_2), \dots$ are presented to the algorithm in sequence.  When each~$(A_t,\theta_t)$ is presented, the algorithm must either ``pass'', or else halt and output ``\,$\E_\rho[A_t] > \theta_t - \eps$''.
        \item If the algorithm passes on all $(A_t, \theta_t)$ pairs, then it ends by outputting ``\,$\E_\rho[A_i] \leq \theta_i$ for all~$i$''.
    \end{itemize}
\end{theorem}

Incidentally, the (offline) quantum Threshold \emph{Decision} problem, where the algorithm only needs to report ``$\exists j : \E_\rho[A_j] > \theta_j - \eps$'' without actually specifying~$j$, is known to be solvable using just $n = O(\log(m)\log(1/\delta)/\eps^2)$ copies~\cite{Aar20}.
We review the proof in \Cref{app:threshold-decision}, tightening/simplifying some quantitative aspects of the underlying theorem of Harrow, Lin, and Montanaro~\cite{HLM17}.
In particular, our tightenings let us slightly improve the copy complexity to $n = O(\log(m/\delta)/\eps^2)$.

\subsubsection{\texorpdfstring{$\chi^2$}{chi-square}-stable threshold reporting}
The most important technical ingredient going into our proof of \Cref{thm:our-gentle-search} is a new, purely classical statistical result fitting into the Adaptive Data Analysis framework (see, e.g.,~\cite{Smi17} for some background).
In that setting one might describe our result as follows: ``adding exponential noise provides a (composably) $\chi^2$-stable mechanism for reporting if a distribution's mean is above a given threshold''.
In more detail, the result says that given a Sample $\bs{S}$ consisting of the sum of $n$ draws from a $\mathrm{Bernoulli}(p)$ distribution (i.e., $\bs{S} \sim \Binomial(n,p)$), if we add independent exponential noise~$\bs{X}$ and then check the event $B$ that $\bs{S} + \bs{X}$ exceeds some large threshold~$\theta n$, then conditioning on~$B$ \emph{not} occurring hardly changes the distribution of~$\bs{S}$, provided $\E[\bs{X}] \gg \stddev[\bs{S}]$.
Here the phrase ``hardly changes'' is in two very strong senses: (i)~we show the random variables $\bs{S} \mid \ol{B}$ and $\bs{S}$ are close even in {$\chi^2$-divergence}, which is a more stringent measure than KL-divergence (or Hellinger distance, or total variation distance) --- that is, the test is ``$\chi^2$-stable''; (ii)~the $\chi^2$-divergence is not just absolutely small, but is even a small fraction of $\Pr[B]^2$ itself (hence the total variation closeness is a small fraction of~$\Pr[B]$).
This allows a kind of composition (as in the ``Sparse Vector'' mechanism~\cite{DR13} from the Differential Privacy literature) in which the same quantum sample can be reused for repeated ``above threshold'' tests, up until the point where having at least one ``above threshold'' outcome becomes likely.
Precisely, our result is the following (refer to
\Cref{sec:prelim-classical} for the
definition and notation of the quantities used below):

\newcommand{\theorestated}{}
\begin{restatable}{theorem}{theostring}\label{thm:claim41} {\normalfont\bfseries\theorestated}
    Let $\bs{S} \sim \Binomial(n,p)$. 
    Assume that $\bs{X}$ is an independent Exponential random variable with mean at least~$\stddev[\bs{S}] = \sqrt{p(1-p)n}$ (and also at least~$1$). 
    Let $B$ be the event that $\bs{S} + \bs{X} > \theta n$, and assume that $\Pr[B] < \frac14$.
    Then
    \[
        \dist[\chi^2]((\bs{S} \mid \ol{B}), \bs{S}) \lesssim 
        \paren*{\Pr[B] \cdot \frac{\stddev[\bs{S}]}{\E[\bs{X}]}}^2 \leq \Pr[B]^2 \cdot (n/\E[\bs{X}]^2).
    \]
\end{restatable}
\noindent (Above we are using the notation $Y \lesssim Z$ to mean $Y \leq C \cdot Z$ for some universal constant~$C$.
We are also abusing notation by writing the  $\chi^2$-divergence between two random variables to mean the $\chi^2$-divergence between their underlying distributions.)
\begin{corollary} \label{cor:claim41}
    Writing $\bs{S}'$ for $\bs{S} \mid \ol{B}$, standard inequalities for $f$-divergences~\cite{GS02} imply
    \[
        \dist[\mathrm{TV}](\bs{S}', \bs{S})
        \leq \dist[\mathrm{H}](\bs{S}', \bs{S})
        \leq \sqrt{\dist[\mathrm{KL}](\bs{S}', \bs{S}) }
        \leq \sqrt{\dist[\chi^2](\bs{S}', \bs{S})}
        \lesssim \Pr[B] \cdot \frac{\stddev[\bs{S}]}{\E[\bs{X}]} \leq  \Pr[B] \cdot \frac{\sqrt{n}}{\E[\bs{X}]}.
    \]
\end{corollary}

Let us remark that our \Cref{thm:claim41} is similar to results appearing previously in the Differential Privacy/Adaptive Data Analysis literature; in particular, it is quite similar to (and inspired by) a theorem (``Claim~41'') of Aaronson and Rothblum~\cite{AR19}.
Although this Claim~41 is presented in a quantum context, the essence of it is a theorem comparable to our \Cref{thm:claim41}, with the following main differences: (i)~it bounds the weaker KL-divergence (though for our applications, this is acceptable); (ii)~the proof is significantly more involved.
(Minor differences include: (i)~their result uses two-sided exponential noise for a two-sided threshold event; (ii)~our bound has the stronger factor $\stddev[\bs{S}]$ instead of just~$\sqrt{n}$.)

\subsubsection{Applications: Shadow Tomography and Hypothesis Selection}
Given our improved \BajanSearch algorithm, we present two applications in quantum data analysis.
The first is to the aforementioned Shadow Tomography problem, where we
obtain a sample complexity that simultaneously achieves the best known
dependence on all three parameters $m$, $d$, and $\eps$.
Furthermore, our algorithm is \emph{online}, as in the Adaptive Data Analysis setting.
\begin{theorem} \label{thm:shadow-tomog}
    There is a quantum algorithm that, given parameters $m \in \NN$, $0 < \eps < \frac12$, and access to unentangled copies of a state~$\rho \in \CC^{d \times d}$, uses
    \[
        \phantom{\qquad \textnormal{($\textsc{l} = \log(\tfrac{\log d}{\delta \eps})$)}}
        n = \frac{(\log^2 m + \textsc{l})(\log d)}{\eps^4} \cdot O(\textsc{l}) \qquad \textnormal{($\textsc{l} = \log(\tfrac{\log d}{\delta \eps})$)}
    \]
    copies of~$\rho$ and then has the following behavior:
    When any (adversarially/adaptively chosen) sequence of observables $A_1, A_2, \dots, A_m \in \CC^{d \times d}$ with $0 \leq A_i \leq \Id$ is presented to the algorithm one-by-one, once $A_t$ is presented the algorithm responds with an estimate~$\wh{\mu}_i$ of $\E_\rho[A_t] = \tr(\rho A_t)$.
    Except with probability at most~$\delta$ (over the algorithm's measurements), all $m$ estimates satisfy $|\wh{\mu}_i - \E_\rho[A_t]| \leq \eps$.
\end{theorem}
The proof of this theorem is almost immediate from our \BajanSearch
algorithm, using a known~\cite{Aar20} black-box reduction to the
mistake-bounded online quantum state learning algorithm of Aaronson, Chen, Hazan, Kale, and Nayak~\cite{ACHKN19}.

Let us philosophically remark that we believe the importance of the parameters, in increasing order, is~$d$, then~$\eps$, then~$m$.
Regarding~$d$, ``in practice'' one may expect that $\log d$, the number of qubits in the unknown state, is not likely to be particularly large.
Indeed, many problems in quantum learning/tomography/statistics~\cite{HM02,CM06,MW16,OW15,OW16,HHJWY17,OW17,AISW19,BOW19,Yu20b,Yu20,BCL20,Yu20c} have polynomial dependence on~$d$, so factors of $\polylog d$ seem of lesser importance.
Regarding~$\eps$, ``in practice'' this might be the most important parameter, as even with a very mild value like $\eps = .1$, a dependence of $1/\eps^4$ is challenging.
It's peculiar that all works on Shadow Tomography have achieved atypical $\eps$-dependence like $1/\eps^4$, $1/\eps^5$, and $1/\eps^8$, instead of the ``expected'' $1/\eps^2$; on the other hand, this peculiarity also seems to occur in the Adaptive Data Analysis literature.
Finally, we feel that the dependence on~$m$ is of the most interest (theoretical interest, at least), and it would be extremely compelling if we could reduce the dependence from $\log^2 m$ to $\log m$.
Our reason is related to quantum Hypothesis Selection, which we now discuss.

\paragraph{Hypothesis Selection.}
The classical (multiple) Hypothesis Selection problem~\cite{Yat85,DL96,DL97} is as follows:
Given are $m$ fixed ``hypothesis'' probability distributions $q_1, \dots, q_m$ on $[d]$, as well as a parameter~$\eps$ and access to samples from an unknown distribution~$p$ on~$[d]$.
The task is to find (with probability at least $1-\delta$) a~$q_j$ which is, roughly, closest to~$p$, while minimizing the number of samples drawn from~$p$.
More precisely, if $\eta = \min_i \{\dist[\mathrm{TV}](p,q_i)\}$, the algorithm should output a hypothesis~$q_j$ with $\dist[\mathrm{TV}](p,q_j) \leq C \eta + \eps$ for some fixed small constant~$C$. 
There are a variety of solutions known to this problem, with standard ones~\cite[Chap.~6]{DL01} achieving $n = O((\log m)/\eps^2)$ (and best constant $C = 3$).
There are also numerous variations, including handling different distance measures besides $\dist[\mathrm{TV}]$~\cite{BV08}, the easier (``realizable/non-robust'') case when $\eta = 0$, and the case when there is a unique answer (as when the hypotheses~$q_j$ are pairwise far apart).
We emphasize that our focus is on the \emph{non-asymptotic regime}, where we would like an explicit sample bound $n = n(m,d,\eps,\delta)$ holding for all values of $m,d,\eps,\delta$.%
\footnote{This is as opposed to the \emph{asymptotic regime}.
There, one focuses on achieving $\delta \leq \exp(-C(m,d,\eps) n)$ for all $n \geq n_0(m,d,\eps)$, where the rate function $C(m,d,\eps)$ should be as large as possible, but where $n_0$ may be a completely uncontrolled function of $m, d, \eps$.  See, e.g.,~\cite{NS11}.}
One particularly useful application of Hypothesis Selection is to \emph{learning} an unknown probability distribution~$p$ from a class~$\mathcal{C}$ (even ``agnostically'').  Roughly speaking, if $\mathcal{C}$ has an \emph{$\eps$-cover} of size $m = m(\eps)$, then one can learn~$p$ to accuracy $O(\eps)$ using a Hypothesis Selection over~$m$ hypotheses; i.e., with $O((\log m)/\eps^2)$ samples in the classical case.
For further discussion of the problem, see e.g.~\cite{ST07}; for Differentially Private Hypothesis Selection, see~\cite{BKSW20,GKKNWZ20}; for fast classical Hypothesis Selection with a quantum computer, see~\cite{QCR20}.

The \emph{quantum} Hypothesis Selection problem is the natural analogue in which probability distributions are replaced by quantum states, and total variation distance is replaced by trace distance.
As with Shadow Tomography (and Differentially Private Hypothesis Selection), it is nontrivial to upgrade classical algorithms due to the fact that samples cannot be naively reused.
We show that one can use Shadow Tomography as a black box to solve quantum Hypothesis Testing.
We also give a different method based on \BajanSearch that achieves an incomparable copy complexity, with a better dependence on~$\eps$ but a worse dependence on~$m$: roughly $(\log^3 m)/\eps^2$, versus the $(\log^2 m)/\eps^4$ of Shadow Tomography.
Finally, we show that if the hypothesis states are pairwise far apart, we can match the optimal bound from the classical case.
\begin{theorem}                                     \label{thm:hypothesis-selection}
    There is a quantum algorithm that, given $m$ fixed hypothesis states $\sigma_1, \dots, \sigma_m \in \CC^{d \times d}$,  parameters $0 < \eps, \delta < \frac12$, and access to unentangled copies of a state~$\rho \in \CC^{d \times d}$, uses
    \[
        n = \min\left\{             \frac{(\log^2 m + \textsc{l}_1)(\log d)}{\eps^4} \cdot O(\textsc{l}_1),
                                \quad
                                            \frac{\log^3 m +
                                              \log(\textsc{l}_2/\delta)
                                              \cdot \log m}{\eps^2}
                                            \cdot O(\textsc{l}_2 \cdot \log(\textsc{l}_2/\delta))
                    \right\}
    \]
    copies of~$\rho$ (where $\textsc{l}_1 = \log(\tfrac{\log d}{\delta \eps})$ and $\textsc{l}_2 = \log(1/{\max\{\eta, \eps\}})$) and has the following guarantee: except with probability at most $\delta$, it outputs~$k$ such that
    \[
        \dist[\mathrm{tr}](\rho, \sigma_k) \leq 3.01 \eta + \eps, \qquad \text{where } \eta = \min_i \{\dist[\mathrm{tr}](\rho,\sigma_i)\}.
    \]

    Further, assuming $\eta < \frac12 (\min_{i \neq j} \{\dist[\mathrm{tr}](\sigma_i,\sigma_j)\} - \eps)$ (so there is a \emph{unique}~$\sigma_i$ near~$\rho$), one can find the $\sigma_k$ achieving $\dist[\mathrm{tr}](\rho, \sigma_k) = \eta$ (except with probability at most~$\delta$) using only $n = O(\log(m/\delta)/\eps^2)$ copies of~$\rho$.
\end{theorem}
The fact that quantum Hypothesis Selection black-box reduces to Shadow Tomography provides significant motivation for trying to prove (or disprove) that Shadow Tomography can be done with $O(\log m) \cdot \poly((\log d)/\eps)$ copies; i.e., that the power on $\log m$ can be reduced to~$1$.
If this were possible, then as in the classical case we would be able to \emph{learn} a quantum state~$\rho \in \CC^{d \times d}$ in a class~$\mathcal{C}$ (to constant trace distance accuracy, say) using $\log (|\mathrm{cover}(\mathcal{C})|)\cdot \polylog(d)$ copies, where $\mathrm{cover}(\mathcal{C})$ denotes a set of states that form a (trace-distance) cover for~$\mathcal{C}$.
It's easy to see that the class $\mathcal{C}$ of \emph{all} states has a
cover of size at most $O(d)^{d^2}$, and hence Shadow Tomography with a $\log m$ dependence would yield a full quantum tomography algorithm with copy complexity $\wt{O}(d^2)$, bypassing the sophisticated representation-theory methods of~\cite{OW16,HHJWY17,OW17}.
One might also hope for more efficient learning of other interesting subclasses of states; e.g., the class \emph{separable} states.

\section{Preliminaries} \label{sec:prelims}

\subsection{Classical probability distributions and distances} \label{sec:prelim-classical}

Let $p = (p_1, \dotsc, p_d)$ denote a probability distribution on $[d] = \set{1, \dotsc, d}$.
We consider $A : [d] \to \RR$ to be a random variable on~$[d]$, and write
\[
    \E_p[A] = \E_{\bs{i} \sim p}[A(\bs{i})] = \sum_{i=1}^d p_i A(i).
\]
In particular, if $A : [d] \to \{0,1\}$ we may think of it as an \emph{event} $A \subseteq [d]$.

Given another probability distribution $q$ on $[d]$, there are a variety of important distances/divergences between $p$~and~$q$.
We now recall all those appearing in \Cref{thm:claim41} and \Cref{cor:claim41}.

The \emph{total variation} distance $\dist[TV](p, q)$ between $p$ and $q$ is defined by
\begin{align*}
  \dist[TV](p, q)
  &= \frac12 \sum_{i=1}^d \abs{p_i - q_i}
    = \max_{A \subseteq [d]}\,\abs*{\E_p[A] - \E_q[A]}.
\end{align*}
The \emph{Bhattacharyya coefficient} $\BC(p, q)$ (an affinity between $p$ and $q$, rather than a distance) is
defined by
\begin{align*}
  \BC(p, q)
  &= \sum_{i=1}^d \sqrt{p_i q_i}.
\end{align*}
This can be used to define \emph{squared Hellinger distance}
$\dist[H](p,q)^2 = \dist[H^2](p, q)$, viz.,
\begin{align*}
  \dist[H^2](p, q)
  &= 2(1 - \BC(p, q))
    = \sum_{i=1}^d \paren*{\sqrt{p_i} - \sqrt{q_i}}^2.
\end{align*}
The \emph{KL-divergence} $\dist[KL](p,q)$ between $p$ and $q$ is defined by
\begin{align*}
  \dist[KL](p, q)
  &= \sum_{i=1}^d p_i \ln(p_i/q_i) = \E_{\bs{i} \sim p} \ln(p_{\bs i}/q_{\bs i}).
\end{align*}
Finally, the \emph{$\chi^2$-divergence} $\dist[\chi^2](p, q)$ between $p$ and
$q$ is defined by
\begin{align*}
  \dist[\chi^2](p, q)
  &= \sum_{i=1}^d q_i \paren*{1 - \frac{p_i}{q_i}}^2
    = \E_{\bs{i} \sim q} \bracket*{\paren*{1 - \frac{p_{\bs i}}{q_{\bs i}}}^2}.
\end{align*}

%
%
\subsection{Quantum states and measurements}

A matrix $A \in \CC^{d \times d}$ is said to be \emph{Hermitian}, or \emph{self-adjoint}, if $A^\dag = A$; here $A^\dag$ denotes the conjugate transpose of $A$.
We write $A \ge 0$ to denote that $A$ is self-adjoint and positive semidefinite; e.g., $B^\dagger B \geq 0$ always. In general, we write $A \geq B$ to mean $A - B \ge 0$.
Recall that a positive semidefinite matrix
$A \ge 0$ has a unique positive semidefinite square root
$\sqrt{A} \ge 0$.
We write $\Id$ for the identity matrix (where the dimension is understood from context).

A $d$-dimensional \emph{quantum state} is any $\rho \in \CC^{d \by d}$ satisfying $\rho \geq 0$ and $\tr \rho = 1$; physically speaking, this is the state of a $d$-level quantum system, such as $\log_2 d$ qubits.
A $d$-dimensional \emph{observable} is any self-adjoint $A \in \CC^{d \times d}$; physically speaking, this is any real-valued property of the system.
One can build an associated measuring device that takes in a quantum system in state~$\rho$, and reads out a (stochastic) real number; we denote its expected value, the \emph{expectation of~$A$ with respect to~$\rho$}, by
\[
    \E_\rho[A] = \tr(\rho A).
\]
It is a basic fact of linear algebra that $\E_\rho[A] \geq 0$ whenever $A \geq 0$.

Note that if $\rho$ and $A$ are diagonal matrices then we reduce to the classical case, where the diagonal elements of~$\rho$ form a probability distribution on~$[d]$ and the diagonal elements of $A$ give a real-valued random variable.

We will use the term \emph{quantum event}\footnote{Also known as a
  \emph{POVM element} in the quantum information literature.} for an
observable $A \in \CC^{d \times d}$ with $0 \leq A \leq \Id$; i.e., a
self-adjoint operator with all its eigenvalues between~$0$ and~$1$.  A
state $\rho \in \CC^{d \times d}$ assigns a probability
$0 \leq \E_\rho[A] \leq 1$ to each event.  We reserve the term
\emph{projector} for the special case when $A^2 = A$; i.e., when all of
$A$'s eigenvalues are either~$0$ or~$1$.  Note that we have not exactly
paralleled the classical terminology, where an ``event'' is a random
variable with all its values equal to~$0$ or~$1$, but: (i)~it's
convenient to have a brief term for observables~$A$ with
$0 \leq A \leq \Id$; (ii)~the terminology ``projector'' is very
standard.  Of course, by the spectral theorem, every quantum event~$A$
may be written as
\begin{equation}    \label{eqn:spectral}
    A = \sum_{i=1}^r \lambda_i \Pi_i, \text{ where each } 0 \leq \lambda_i \leq 1, \text{ and } \Pi_i\text{'s are pairwise orthogonal projectors.}
\end{equation}

%
%

A \emph{quantum measurement} $\cM$, also known as a positive-operator
valued measure (POVM), is a sequence $\cM = (A_1, \dots, A_k)$ of
quantum events with $A_1 + \dotsb + A_k = \Id$.
Since
\begin{align*}
    \E_\rho[A_1] + \cdots + \E_\rho[A_k]
  &= \E_\rho[A_1 + \dotsb + A_k]
    = \E_\rho[\Id] = 1,
\end{align*}
a state $\rho$ and a measurement $\cM$ determine a probability
distribution $p$ on $[k]$ defined by $p_i = \E_\rho[A_i]$ for
$i = 1, \dotsc, k$.
A common scenario is that of a \emph{two-outcome measurement}, associated to any quantum event~$A$; this is the measurement $\cM = (\ol{A}, A)$, where $\ol{A} = \Id - A$.

For any quantum measurement~$\cM$, one can physically implement a measuring device that, given~$\rho$, reports $\bs{i} \in [k]$ distributed according to~$p$.
Mathematically, an \emph{implementation} of $\cM = (A_1, \dots, A_k)$ is a sequence of $d$-column matrices $M_1, \dotsc, M_k$ with $M_i^\dagger M_i = A_i$ for $i = 1, \dotsc, k$.
Under this implementation, conditioned on the readout being $\bs{i} = i$, the state~$\rho$ collapses to the new state~$\restrict{\rho}_{M_i}$, defined as follows:
\[
      \restrict{\rho}_{M_i} =  \frac{M_i \rho M_i^\dagger}{\E_\rho[M_i^\dagger M_i]} = \frac{M_i \rho M_i^\dagger}{\E_\rho[A_i]}.
\]
Given $\cM$, we will define the \emph{canonical implementation} to be the one in which $M_i = \sqrt{A_i}$.
In particular, if we have any quantum event~$A$ and we canonically implement the associated two-outcome measurement~$(\ol{A}, A)$, then measuring $\rho$ and conditioning on $A$~occurring yields the new state
\[
      \restrict{\rho}_{\sqrt{A}} =  \frac{\sqrt{A} \rho \sqrt{A}}{\E_\rho[A]}.
\]

More generally, we have the notion of a \emph{quantum operation}~$S$ on $d$-dimensional states, defined by $d$-column matrices $M_1, \dotsc, M_k$
such that
\begin{align*}
  M_1^\dag M_1 + \dotsb + M_k^\dag M_k \le \Id.
\end{align*}
The result of applying $S$ to a state~$\rho$ is (the sub-normalized state)
\begin{align*}
  S(\rho)
  &= M_1 \rho M_1^\dag + \dotsb + M_k \rho M_k^\dag.
\end{align*}
An operation $S$ defines a measurement
\begin{align*}
\cM_S
&= (M_1^\dag M_1, M_2^\dag
M_2, \dotsc, M_k^\dag M_k, \unit - (M_1^\dag M_1 + \dotsb + M_k^\dag M_k)).
\end{align*}
In \Cref{sec:damage-lemma} below, we will use the following terminology:
we say a quantum operation $S$ \emph{rejects} a state $\rho$ if the
outcome of measuring $\rho$ according to $\cM_S$ corresponds to the
quantum event $\unit - (M_1^\dag M_1 + \dotsb + M_k^\dag M_k)$;
otherwise, we say $S$ \emph{accepts} $\rho$.

%
%
\cnote{Here we use bra-ket notation without introducing it.}\rnote{I think it's fine.  You can change it to $\begin{pmatrix} 1 & 0 \\ 0 & 0 \end{pmatrix}$ if you prefer.}
Finally, we will use the following special case of the well-known
Naimark dilation theorem:
\begin{theorem}[Naimark]\label{thm:naimark-extension}
  If $A \in \CC^{d \times d}$ is a quantum event, then there exists a
  projector $\Pi$ operating on the space $\CC^{2d}$ such
  that, for any $\rho \in \CC^{d \times d}$,
  \begin{align*}
    \E_{\rho \tensor \ketbra{0}{0}}[\Pi] = \E_{\rho}[A].
  \end{align*}
\end{theorem}

\subsection{Quantum state distances}
Just as with classical probability distributions, there are a variety of distances/divergences between two quantum states $\rho, \sigma \in \CC^{d \times d}$.
In fact, for every classical ``$f$-divergence'' there is a corresponding ``measured quantum $f$-divergence'', which is the maximal classical divergence that can be achieved by performing the same measurement on~$\rho$ and~$\sigma$.
In this way, classical total variation distance precisely corresponds to quantum trace distance, the Bhattacharyya coefficient precisely corresponds to quantum fidelity, etc.
See, e.g.,~\cite[Sec.~3.1.2]{BOW19} for further review; here we will simply directly define some quantum distances.

The \emph{trace distance} $\dist[\tr](\rho, \sigma)$ between states
$\rho$ and $\sigma$ is defined by
\begin{align*}
  \dist[\tr](\rho, \sigma)
  &= \half \norm{\rho - \sigma}_1
    = \max_{0 \le A \le \Id}\, \abs{\E_\rho[A] - \E_\sigma[A]}.
\end{align*}
Here the second equality is known as the Holevo--Helstrom
theorem~\cite{Holevo:1973,Hel76}, and the maximum is over all quantum events $A \in
\CC^{d \by d}$. Moreover, the maximum is achieved by a projector.
The \emph{fidelity} $\Fid(\rho, \sigma)$ between states $\rho$ and $\sigma$ is defined by
\begin{align}
  \Fid(\rho, \sigma)
  &= \norm{\sqrt{\rho} \sqrt{\sigma}}_1 = \tr \sqrt{\sqrt{\rho} \sigma \sqrt{\rho}}. \label{eqn:fid}
\end{align}
This can be used to define the \emph{squared Bures distance} $\dist[Bures](\rho,\sigma)^2 = \dist[Bures^2](\rho,\sigma)$, viz.,
\begin{align*}
  \dist[Bures^2](\rho, \sigma)
  &= 2(1 - \Fid(\rho, \sigma)).
\end{align*}
It follows from the work of Fuchs and Caves~\cite{FC95} that
$\half \dist[Bures^2](\rho, \sigma) \le \dist[tr](\rho, \sigma) \le
\dist[Bures](\rho, \sigma)$ for all states $\rho$ and $\sigma$.

Below we give a simpler formula for fidelity in the case when $\sigma$ is a conditioned version of~$\rho$ (such results are sometimes known under the name ``gentle measurement''; see~\cite[Cor.~3.15]{Watrous:2018}):
\begin{proposition}                                       \label{prop:fid1}
    Let $\rho \in \CC^{d \times d}$ and $M \in \CC^{d \times d}$ an
    observable.
    Then
    $\displaystyle
        \Fid(\rho, \restrict{\rho}_{M})^2 = \frac{\E_\rho[M]^2}{\E_\rho[M^2]}.
    $
    In particular, for a projector~$\Pi$ we get
    $
        \Fid(\rho,  \restrict{\rho}_{\Pi}) = \sqrt{\E_{\rho}[\Pi]},
    $
    and for conditioning on the occurrence of a quantum event~$A$ (under the canonical implementation),
    $\displaystyle
        \Fid(\rho, \restrict{\rho}_{\sqrt{A}}) = \frac{\E_\rho[\sqrt{A}]}{\sqrt{\E_\rho[A]}}.
    $
\end{proposition}
\begin{proof}
    Using the definition of $\restrict{\rho}_{M}$ and the second formula for fidelity in \Cref{eqn:fid},
    \[
        \Fid(\rho, \restrict{\rho}_{M})^2 = \frac{\tr\paren*{\sqrt{\sqrt{\rho} M \rho M^\dagger \sqrt{\rho}}}^2}{\E_{\rho}[M^\dagger M]} = \frac{\tr\paren*{\sqrt{\sqrt{\rho} M \sqrt{\rho} \sqrt{\rho} M \sqrt{\rho}}}^2}{\E_{\rho}[M^2]} = \frac{\tr\paren*{\sqrt{\rho} M \sqrt{\rho}}^2}{\E_{\rho}[M^2]}
        = \frac{\E_{\rho}[M]^2}{\E_{\rho}[M^2]}. \qedhere
    \]
\end{proof}
Below we give a further formula for $\Fid(\rho, \restrict{\rho}_{\sqrt{A}})$ using the spectral decomposition of~$A$.  (We remark that it may be obtained as a special case of the theorem of Fuchs and Caves~\cite{FC95}.)
\begin{proposition}\label{prop:fidelity-bhattacharyya}
  Let $\rho \in \CC^{d \times d}$ be a quantum state, and let $A \in \CC^{d \times d}$ be a quantum event with spectral decomposition $A = \sum_{i=1}^r \lambda_i \Pi_i$ as in \Cref{eqn:spectral}.
  Let $p$ be the probability distribution on~$[r]$ determined by measurement $\cM = (\Pi_1, \dots, \Pi_r)$ on~$\rho$, and let $q$ be the one determined by $\cM$ on $\restrict{\rho}_{\sqrt{A}}$.  Then
  $\Fid(\rho, \restrict{\rho}_{\sqrt{A}}) = \BC(p,q)$.
\end{proposition}
\begin{proof}
    By definition,
    \[
        \E_{\restrict{\rho}_{\sqrt{A}}}[\Pi_i] \cdot \E_{\rho}[A] = \tr(\sqrt{A} \rho \sqrt{A} \Pi_i) = \E_{\rho}[\sqrt{A} \Pi_i \sqrt{A}] = \E_{\rho}[\lambda_i \Pi_i] = \lambda_i p_i,
    \]
    and hence $q_i = \lambda_i p_i / \E_{\rho}[A]$.
    It follows that
    \[
        \BC(p,q) = \frac{\sum_i \sqrt{\lambda_i} p_i}{\sqrt{\E_{\rho}[A]}} = \frac{\E_{\rho}[\sqrt{A}]}{\sqrt{\E_{\rho}[A]}},
    \]
    and the proof is complete by \Cref{prop:fid1}.
\end{proof}

\subsection{Naive expectation estimation}
\begin{lemma}\label{lem:amplification}
  Let $E \in \CC^{d \times d}$ be a quantum event and let
  $0 < \eps, \delta < \frac12$. Then there exists $n = O(\log(1/\delta)/\eps^2)$ (not depending on~$E$)
  and a measurement $\cM = (A_0, \dotsc, A_n)$ such that, for any quantum state
  $\rho \in \CC^{d \times d}$,
  \begin{align*}
    \Pr \bracket*{\abs*{\frac{\bs{k}}{n} -
    \tr(\rho E)} > \eps} \le \delta,
  \end{align*}
  where $\bs{k} \in \set{0, \dotsc, n}$ is the random outcome of the
  measurement $\cM$ applied to the state $\rho^{\tensor n}$.

  Moreover, for any parameters $0 \le \tau, c \le 1$, there exists a quantum
  event $B$ such that
  \begin{align*}
    \abs{\tr(\rho E) - \tau} > c + \eps
    &\implies \E_{\rho^{\tensor n}}[B] \ge 1 - \delta \ \text{and} \\
    \abs{\tr(\rho E) - \tau} \le c - \eps
    &\implies \E_{\rho^{\tensor n}}[B] \le \delta.
  \end{align*}
  Additionally, if $E$ is a projector, then so is $B$.
\end{lemma}
\begin{proof}
  Let $E_1 = E$ and $E_0 = \Id - E$. For all $x \in \Bin[n]$, let
  $E_x \in (\CC^{d \times d})^{\tensor n}$ be defined by
  $E_x = E_{x_1} \tensor E_{x_2} \tensor \dotsb \tensor E_{x_n}$. For
  $k = 0, \dotsc, n$, let $A_k \in (\CC^{d \times d})^{\tensor n}$ be
  the quantum event defined by
  \begin{align*}
    A_k
    &= \sum_{\substack{x \in \Bin[n] \\ \abs{x} = k}} E_x.
  \end{align*}
  Let $\cM$ be the measurement defined by
  $\cM = \set{A_0, \dotsc, A_n}$.

  Thus, if $\bs{k} \in \set{0, \dotsc, n}$ is the random outcome of
  measuring $\rho^{\tensor n}$ according to $\cM$, then $\bs{k}$ is
  distributed as $\Binomial(n, \tr(\rho E))$. Hence, if
  $n = O(\log(1/\delta)/\eps^2)$, then, by Hoeffding's inequality,
  \begin{align*}
    \Pr \bracket*{\abs*{\frac{\bs k}{n} -
    \tr(\rho E)} \ge \eps}
    &\le 2 \exp(- 2n\eps^2)
      \le \delta.
  \end{align*}
  Let parameters $\tau, c \in [0, 1]$ be given and let the function $f : [0, 1]
  \to \Bin$ be defined by
  \begin{align*}
    f(t) &=
    \begin{cases}
      1, & \abs{t - \tau} \ge c, \\
      0, & \text{otherwise}.
    \end{cases}
  \end{align*}
  Finally, let the quantum event $B$ be defined by
  \begin{align*}
    B
    &= \sum_{k=0}^n f(k/n) A_k.
  \end{align*}
  Thus, if $\bs{k} \sim \Binomial(n, \tr(\rho E))$, then
  \begin{align*}
    \E_{\rho^{\tensor n}} [B]
    &= \sum_{k=0}^n \Pr[\bs{k} = k] \cdot f(k/n)
      = \E[f(\bs{k}/n)]
      = \Pr \bracket*{\abs*{\frac{\bs{k}}{n} - \tau} \ge c}.
  \end{align*}
  If $c + \eps \le \abs{\tr(\rho E) - \tau}$, then
  $\abs{\tr(\rho E) - \bs{k}/n} < \eps$ implies
  $\abs{\bs{k}/n - \tau} \ge c$. Hence,
  \begin{align*}
    \E_{\rho^{\tensor n}} [B]
    &= \Pr \bracket*{\abs*{\frac{\bs{k}}{n} - \tau} \ge c}
      \ge \Pr \bracket*{\abs*{\frac{\bs{k}}{n} - \tr(\rho E)} < \eps}
      \ge 1 - \delta.
  \end{align*}
  If $c - \eps \ge \abs{\tr(\rho E) - \tau}$, then $\abs{\tr(\rho E) -
    \bs{k}/n} < \eps$ implies $\abs{\bs{k}/n - \tau} < c$. Hence,
  \begin{align*}
    \E_{\rho^{\tensor n}} [\ol{B}]
    &= \Pr \bracket*{\abs*{\frac{\bs{k}}{n} - \tau} < c}
      \ge \Pr \bracket*{\abs*{\frac{\bs{k}}{n} - \tr(\rho E)} < \eps}
      \ge 1 - \delta.
  \end{align*}
  If $E$ is a projector, then $A_k$ is a projector and
  $A_k A_\ell = A_\ell A_k = 0$ for all
  $k, \ell \in \set{0, \dotsc, n}$. Since $B$ is a sum of orthogonal
  projectors $A_k$ with $k \in \set{0, \dotsc, n}$, it follows that $B$
  is a projector.
\end{proof}

\subsection{Quantum union bound-style results}\label{sec:damage-lemma}
The following result is part of the ``Damage Lemma'' of Aaronson and
Rothblum~\cite[Lemma 17]{AR19}. Since the original proof of the ``Damage
Lemma'' was found to be incorrect~\cite{Lei22}, we provide a slightly
different proof by induction below:
\begin{lemma}\label{lem:damage}
  Let $S_1, \dotsc, S_m$ be arbitrary quantum operations on $d$-dimensional
  quantum states. Let $\rho$ be a quantum state on $\CC^d$ with
  $p_i = \tr(S_i(\rho)) > 0$ for all $i \in [m]$. It holds that
  \begin{align*}
    \abs{\tr(S_m(\dotsm S_1(\rho))) - p_1 \dotsm p_m}
    &\le 2 \cdot \sum_{k=1}^{m-1} p_1 \dotsm p_k \cdot
      \dist[\tr] \paren*{\frac{S_k(\rho)}{\tr(S_k(\rho))}, \rho}.
  \end{align*}
\end{lemma}
\begin{proof}
  For all $k \in [m]$, let $p_{[k]} = p_1 \dotsm p_k$ and
  $\sigma_k = S_k(\rho)/\tr(S_k(\rho))$. For all self-adjoint matrices
  $X$, $\abs{\tr(X)} \le \norm{X}_1$ and $\norm{S(X)}_1 \le \norm{X}_1$
  for all quantum operations $S$. Hence,
  \begin{align*}
    \abs{\tr(S_m(\dotsb S_1(\rho))) - p_{[m]}}
    &= \abs{\tr(S_m(\dotsb S_1(\rho))) - p_{[m-1]}
      \tr(S_m(\rho))} \\
    &= \abs{\tr(S_m(\dotsb S_1(\rho)) - p_{[m-1]}
      S_m(\rho))} \\
    &= \abs{\tr(S_m(S_{m-1}(\dotsb S_1(\rho)) - p_{[m-1]}
      \rho))} \\
    &\le \norm{S_m(S_{m-1}(\dotsb S_1(\rho)) - p_{[m-1]}
      \rho)}_1 \\
    &\le \norm{S_{m-1}(\dotsb S_1(\rho)) - p_{[m-1]}
      \rho}_1 \\
    &\le \norm{S_{m-1}(\dotsb S_1(\rho)) - p_{[m-1]} \sigma_{m-1}}_1 +
      \norm{p_{[m-1]} \sigma_{m-1} - p_{[m-1]} \rho}_1 \\
    &= \norm{S_{m-1}(\dotsb S_1(\rho)) - p_{[m-2]} S_{m-1}(\rho)}_1 +
      2 p_{[m-1]} \dist[tr](\sigma_{m-1}, \rho) \\
    &\le \norm{S_{m-2}(\dotsb S_1(\rho)) - p_{[m-2]} \rho}_1 +
      2 p_{[m-1]} \dist[tr](\sigma_{m-1}, \rho).
  \end{align*}
  Note that
  $\norm{S_1(\rho) - p_1 \rho}_1 = p_1 \norm{\sigma_1 - \rho}_1 = 2
  p_{[1]} \dist[tr](\sigma_1, \rho)$. Therefore, by induction,
  \begin{align*}
    \abs{\tr(S_m(\dotsb S_1(\rho))) - p_{[m]}}
    &\le 2 \cdot \sum_{k=1}^{m-1} p_{[k]} \cdot \dist[tr](\sigma_k,
      \rho). \qedhere
  \end{align*}
\end{proof}

\Cref{lem:damage} compares the probability
$\tr(S_1(\rho)) \dotsm \tr(S_m(\rho))$ that the operations
$S_1, \dotsc, S_m$ accept the same state $\rho$ independently with the
probability $\tr(S_m(\dotsm S_1(\rho)))$ that all $S_1, \dotsc, S_m$
accept when applied sequentially to the initial state $\rho$.

The following inequality, which appears in the proof of~\cite[Theorem
1.3]{OV22}, will be used to show that when $S_1, \dotsc, S_m$ are
applied sequentially to the initial state $\rho$, the probability of
observing $S_1, \dotsc, S_{t-1}$ accept and $S_t$ reject for certain
``good'' values of $t \in [m]$ is bounded below by a positive constant
for specific $\rho$ and $S_1, \dotsc, S_m$ (see proof of
\Cref{lem:bajansearch3}).

\begin{lemma}\label{lem:fid-ineq}
  Let $\rho$ be a mixed quantum state and let $A_1, \dotsc, A_m$ denote
  quantum events on $\CC^d$ with $\E_\rho[A_i] > 0$ for all $i \in
  [m]$. Let $p_0 = 1$, $q_0 = 1$, $\rho_0 = \rho$, $p_i = 1 - \E_\rho[A_i]$, and
  $\rho_i = \restrict{\rho_{i-1}}_{\sqrt{A_i}}$ for all $i \in
  [m]$.

  Suppose the measurements $(A_1, \ol{A}_1), \dotsc, (A_m, \ol{A}_m)$
  are applied to $\rho$ sequentially; for all $t \in [m]$, let $q_t$ denote the probability
  of observing outcomes $A_1, \dotsc, A_t$ and let $s_t$ denote the
  probability of observing outcomes $A_1, \dotsc, A_{t-1}, \ol{A}_t$. It
  holds that
  \begin{align*}
    1
    &\le \sqrt{q_m} \Fid(\rho, \rho_m) + \sum_{i=1}^m \sqrt{s_i} \sqrt{p_i}.
  \end{align*}
  \begin{proof}
    Since $1 = q_0 \Fid(\rho, \rho_0)$ and $q_i = q_{i-1} \cdot \E_{\rho_{i-1}}[A_i]$ for all $i \in [m]$,
    \begin{align*}
      1 - \sqrt{q_m} \Fid(\rho, \rho_m)
      &= \sum_{i=1}^m \paren*{\sqrt{q_{i-1}} \Fid(\rho, \rho_{i-1}) -
        \sqrt{q_{i}} \Fid(\rho, \rho_{i})} \\
      &= \sum_{i=1}^m \paren*{\sqrt{q_{i-1}} \Fid(\rho, \rho_{i-1}) -
        \sqrt{q_{i-1}} \sqrt{\E_{\rho_{i-1}}[A_i]} \Fid(\rho, \rho_{i})}
      \\
      &= \sum_{i=1}^m \sqrt{q_{i-1}} \paren*{\Fid(\rho, \rho_{i-1}) -
        \sqrt{\E_{\rho_{i-1}}[A_i]} \Fid(\rho, \rho_{i})}.
    \end{align*}
    By~\cite[Lemma 2.1]{OV22} and the inequality $\unit - \sqrt{A_i} \le \ol{A_i}$,
    \begin{align*}
      \Fid(\rho, \rho_{i-1}) - \sqrt{\E_{\rho_{i-1}}[A_i]} \Fid(\rho,
      \rho_{i})
      &\le \sqrt{\E_{\rho}[\unit - \sqrt{A_i}]}
        \sqrt{\E_{\rho_{i-1}}[\unit - \sqrt{A_i}]}
        \le \sqrt{\E_{\rho}[\ol{A}_i]} \sqrt{\E_{\rho_{i-1}}[\ol{A}_i]}.
    \end{align*}
    Hence,
    \begin{align*}
      1 - \sqrt{q_m} \Fid(\rho, \rho_m)
      &\le \sum_{i=1}^m \sqrt{q_{i-1}} \sqrt{\E_{\rho}[\ol{A}_i]}
        \sqrt{\E_{\rho_{i-1}}[\ol{A}_i]}
        \le \sum_{i=1}^m \sqrt{s_i}
        \sqrt{p_i}.
    \end{align*}
  \end{proof}
\end{lemma}

Finally, for the ``unique decoding'' part of our Hypothesis Selection routine we will use a related result, Gao's \emph{quantum Union Bound}~\cite{Gao15}:
\begin{lemma}   \label{lem:quantum-union-bound}
    For each of $i = 1,  \dots,  m$, let $\Pi^1_i \in \CC^{d \by d}$ be a projector and write $\Pi^0_i = \Id - \Pi^1_i$.
    Then for any quantum state~$\rho \in \CC^{d \by d}$,
    \[
        \E_\rho[(\Pi^1_1 \dotsm \Pi^1_m)(\Pi^1_1 \dotsm \Pi^1_m)^\dag] \geq 1 - 4 \sum_{i=1}^m \E_\rho[\Pi^0_i].
    \]
\end{lemma}
\begin{corollary}                                       \label{cor:quantum-union-bound}
    In the setting of \Cref{lem:quantum-union-bound}, suppose that $x \in \{0,1\}^m$ is such that $\E_\rho[\Pi_i^{x_i}] \geq 1 - \eps$ for all $1 \leq i \leq m$.
    If an algorithm sequentially measures $\rho$ with $(\Pi_1^0, \Pi_1^1)$, measures the resulting state with $(\Pi_2^0, \Pi_2^1)$, measures the resulting state with $(\Pi_3^0, \Pi_3^1)$, etc., then the probability that the measurement outcomes are precisely $x_1, x_2, \dots, x_m$ is at least $1 - 4\eps m$.
\end{corollary}

\section{\texorpdfstring{$\chi^2$}{chi-square}-stable Threshold Reporting} \label{sec:exponential}

Our goal in this section is to prove \Cref{thm:claim41} and to show how
this classical result applies to quantum states and measurements. We
begin with some preparatory facts.

The following is well known~\cite{Ber95}:
\begin{proposition}\label{prop:variance}
    For $\bs{S}$ a random variable and $f : \RR \to \RR$  1-Lipschitz,  $\Var[f(\bs{S})] \leq \Var[\bs{S}]$.
\end{proposition}
\begin{proof}
    Let $\bs{S}'$ be an independent copy of $\bs{S}$.
    Since the function $f$ is 1-Lipschitz, we always have $\tfrac12(f(\bs{S}) - f(\bs{S}'))^2 \leq \tfrac12(\bs{S} - \bs{S}')^2$.
    The result follows by taking expectations of both sides.
\end{proof}

We will also use the following simple numerical inequality:
\begin{lemma}\label{lem:inequality-1}
  Fix $0 \leq p \leq 1$,  $q = 1-p$. Then for $C = (e-1)^2 \leq 3$, we have
  \[
    \phantom{\quad \forall \lambda \in [0,1].}
        q + pe^{2 \lambda} \leq (1 + Cpq \lambda^2) \cdot (q + pe^{\lambda})^2
    \quad \forall \lambda \in [0,1].
  \]
\end{lemma}
\begin{proof}
    Since $(q+pe^{\lambda})^2 \geq (q+p)^2 = 1$ for $\lambda \geq 0$, it suffices to show
    \[
        q + pe^{2 \lambda} \leq (q + pe^{\lambda})^2 + Cpq \lambda^2 \quad \forall \lambda \in [0,1].
    \]
    But $(q + p\Lambda^2) - (q + p\Lambda)^2 = pq(\Lambda-1)^2$ when
    $p+q = 1$, so it is further equivalent to show
    \[
        (e^{\lambda} - 1)^2 \leq C\lambda^2 \quad \forall \lambda \in [0,1].
    \]
    But this indeed holds with $C = (e-1)^2$, as it is equivalent to $e^{\lambda} \leq 1 + (e-1) \lambda$ on $[0,1]$, which follows from convexity of $\lambda \mapsto e^{\lambda}$.
\end{proof}

We now do a simple calculation showing how much a random variable changes (in $\chi^2$-divergence) when conditioning on an event.
In using the below, the typical mindset is that~$B$ is an event that ``rarely'' occurs, so $\Pr[\overline{B}]$ is close to~$1$.
\begin{proposition} \label{prop:chi-squared}
  Let $\bs{S}$ be a discrete random variable, and let $B$ be an event on the same probability space with $\Pr[B] <
  1$. For each outcome~$s$ of~$\bs{S}$, define $f(s) = \Pr[B \mid \bs{S} = s]$. Then
  \[
    \dist[\chi^2]((\bs{S} \mid \ol{B}), \bs{S})
    = \Var[f(\bs{S})] \bigm/ \Pr[\ol{B}]^2.
  \]
\end{proposition}
\begin{proof}
  We have the likelihood ratio $\Pr[\bs S = s \mid \ol{B}]\bigm/\Pr[\bs S = s] = (1 - f(s))\bigm/\Pr[\ol{B}]$, by Bayes' theorem. Hence,
  \[
    \dist[\chi^2]((\bs{S} \mid \ol{B}), \bs{S})
    = \E \bracket*{\paren*{1 - \frac{1 - f(\bs S)}{\Pr[\ol{B}]}}^2}
    = \frac{1}{\Pr[\ol{B}]^2} \E \bracket*{(f(\bs S) - \Pr[B])^2}
    = \Var[f(\bs S)] \bigm/ \Pr[\ol{B}]^2,
  \]
  where the last step uses $\E[f(\bs S)] = \Pr[B]$.
\end{proof}

We can now prove \Cref{thm:claim41}, which we restate for convenience:

\renewcommand{\theorestated}{(Restated)\ }
\theostring*

\begin{proof}
    Write $\lambda = 1/\E[\bs{X}]$, so $\bs{X} \sim \Exponential(\lambda)$ and we have the assumptions $\lambda \leq \frac{1}{\sqrt{pqn}}$ and $\lambda \leq 1$.
    Using \Cref{prop:chi-squared} and $\Pr[\ol{B}] > \frac34$, it suffices to show
    \[
        \Var[f(\bs{S})] \lesssim \Pr[B]^2 \cdot pq n \lambda^2,
    \]
    where
    \[
        f(s) = \Pr[\bs{X} > \theta n - s] = \min(1, g(s)), \qquad g(s) =  \exp(-\lambda(\theta n - s)).
    \]
    Since $y \mapsto
    \min(1,y)$ is 1-Lipschitz, \Cref{prop:variance} tells us that
    $\Var[f(\bs{S})] \leq \Var[g(\bs{S})]$.
    $\Var[g(\bs{S})]$ can be computed using the moment-generating
    function of $\bs{S} \sim \Binomial(n,p)$, namely $\E[\exp(t \bs{S})]
    = (q + pe^t)^n$:
    \begin{align*}
        \E[g(\bs S)] &= \E[\exp(- \lambda(\theta n - \bs S))]
                            = \exp(- \lambda \theta n) \cdot (q + pe^{\lambda})^n, \\
        \E[g(\bs S)^2] &= \E[\exp(- 2 \lambda(\theta n - \bs S))]
                            = \exp(- 2 \lambda \theta n) \cdot (q + pe^{2 \lambda})^n.
    \end{align*}
    Thus
    \begin{align*}
        \Var[g(\bs S)] 
         = \E[g(\bs S)]^2 \cdot \paren*{\frac{\E[g(\bs S)^2]}{\E[g(\bs S)]^2} - 1}
    &=  \E[g(\bs S)]^2 \cdot \paren*{\paren*{\frac{q + pe^{2 \lambda}}{(q + pe^{\lambda})^2}}^n - 1} \\
    &\le\E[g(\bs S)]^2 \cdot \paren*{(1+3pq\lambda^2)^n - 1} \tag{\Cref{lem:inequality-1}}\\
    &\lesssim \E[g(\bs S)]^2 \cdot pqn \lambda^2 \tag{as $\lambda^2 \leq \frac{1}{pqn}$}
    \end{align*}
    and it therefore remains to establish
    \begin{equation}    \label[ineq]{eqn:bugfix}
        \E[g(\bs S)] = \exp(- \lambda \theta n) \cdot (q + pe^{\lambda})^n \lesssim \Pr[B].
    \end{equation}
    Intuitively this holds because $g(s)$ should not be much different from $f(s)$, and $\E[f(\bs S)] = \Pr[B]$ by definition.
    Formally, we consider two cases: $p \geq \frac1n$ (intuitively, the main case) and $p \leq \frac1n$.

    \paragraph{Case 1:} $p \geq \frac1n$.  In this case we use that $\Pr[\bs{S} > pn] \geq \frac14$ (see, e.g.,~\cite{Doe18}), and hence: (i)~it must be that $\theta \geq p$, since we are assuming $\Pr[B] = \Pr[\bs{S} + \bs{X} > \theta n] < \frac14$; and, (ii)~$\Pr[B] \geq \Pr[\bs{S} > pn] \cdot \Pr[\bs{X} \geq (\theta-p)n] \geq \frac14 \exp(-\lambda(\theta-p)n)$, where the first inequality used independence of $\bs{S}$ and~$\bs{X}$ and the second inequality used $(\theta -p)n \geq 0$ (by~(i)).
    Thus, to establish \Cref{eqn:bugfix}, it remains to show $\exp(-
    \lambda \theta n) \cdot (q + pe^{\lambda})^n \lesssim
    \exp(-\lambda(\theta-p)n)$.

    Since $0 < \lambda \le 1$,
    \begin{align*}
      e^\lambda - 1
      &= \sum_{i \ge 1} \frac{\lambda^i}{i!}
        = \lambda + \lambda^2 \sum_{i \ge 2} \frac{\lambda^{i-2}}{i!}
        \le \lambda + \lambda^2 \sum_{i \ge 2} \frac{1}{i!}
        \le \lambda + \lambda^2 e.
    \end{align*}
    By a similar argument,
    $e^{-\lambda} - 1 \le - \lambda + \lambda^2 e$. Using these two
    inequalities and $1 + x \le e^x$ for $x \in \RR$, we obtain
    \begin{align*}
      (q + p e^\lambda)^n
      &= (1 + p(e^\lambda - 1))^n
        \le \exp(p (e^\lambda - 1) n)
        \le \exp(\lambda p n) \exp(e \lambda^2 \cdot p \cdot n) \qquad \text{and} \\
      (q + p e^\lambda)^n
      &= \exp(\lambda n) (p + q e^{-\lambda})^n
        = \exp(\lambda n) (1 + q (e^{-\lambda} - 1))^n \\
      &\le \exp(\lambda n) \exp(q (e^{-\lambda} - 1) n)
        \le \exp(\lambda p n) \exp(e \lambda^2 \cdot q \cdot n).      
    \end{align*}
    Hence, $(q + p e^\lambda)^n \le \exp(\lambda p n) \exp(e
    \lambda^2 \cdot \min\set{p, q} \cdot n)$. Since, $\lambda^2 \le 1/pqn$, by
    assumption, it follows that $\lambda^2 \min\set{p, q}n \le
    1/\max\set{p, q} \le 2$, so
    \begin{align*}
      (q + p e^\lambda)^n
      &\le \exp(\lambda p n) \exp(e / \max\set{p, q})
        \le \exp(\lambda p n) \exp(2e).
    \end{align*}
    Therefore,
    $
    \exp(- \lambda \theta n) \cdot (q + pe^{\lambda})^n
    \lesssim \exp(- \lambda \theta n) \exp(\lambda p n)
    = \exp(- \lambda (\theta - p) n)
    $,
    as needed.
    
        

    \paragraph{Case 2:} $p \leq \frac1n$. Since $\lambda \in (0, 1]$, we
    have $e^\lambda \le 1 + 2 \lambda$. Hence,
    $q + pe^{\lambda}\leq 1 + 2p\lambda \leq 1 + \frac{2}{n}$, and so
    $(q+pe^{\lambda})^n \lesssim 1$, meaning that \Cref{eqn:bugfix}
    follows from
    $\Pr[B] \geq \Pr[\bs{X} > \theta n] = \exp(-\lambda \theta n)$.
\end{proof}

\subsection{The quantum version}
Having established~\Cref{thm:claim41}, we now show how this result
applies to quantum states and measurements. Specifically, we prove that
for any quantum event $A \in \CC^{d \times d}$, there exists a
corresponding event $B \in (\CC^{d \times d})^{\tensor n}$ which
exhibits the same statistics as the classical event
$\bs{S} + \bs{X} > \theta n$ from~\Cref{thm:claim41} with
$\bs{S} \sim \Binomial(n, \tr(\rho A))$ when $\rho^{\tensor n}$ is
measured according to $B$. Moreover, we also relate the fidelity between
the states $\rho^{\tensor n}$ and
$\restrict{\rho^{\tensor n}}_{\sqrt{\Id - B}}$ (i.e.\ the state
$\rho^{\tensor n}$ conditioned on the event $\Id - B$) to the
Bhattacharyya coefficient between $\bs{S}$ and
$(\bs{S} \mid \bs{S} + \bs{X} \le \theta n)$ (i.e.\ $\bs{S}$ conditioned
on the event $\bs{S} + \bs{X} \le \theta n$).

\begin{lemma}\label{lem:meas-claim41}
  Let $\rho \in \CC^{d \times d}$ represent an unknown quantum state and
  let $A \in \CC^{d \times d}$ be a projector. Let $n \in \NN$, let
  $\lambda > 0$, and let $\theta \in [0, 1]$ be an arbitrary
  threshold. Let $\bs{S}$ and $\bs{X}$ be classical random variables
  with distributions defined by $\bs{S} \sim \Binomial(n, \E_{\rho}[A])$
  and $\bs{X} \sim \Exponential(\lambda)$. There exists a quantum event
  $B \in (\CC^{d \times d})^{\tensor n}$ such that
  $\E_{\rho^{\tensor n}}[B] = \Pr[\bs{S} + \bs{X} > \theta n]$ and
  \begin{align*}
    \Fid \paren*{\rho^{\tensor n}, \restrict{\rho^{\tensor
    n}}_{\sqrt{\Id - B}}}
    &= \BC((\bs{S} \mid \bs{S} + \bs{X} \le \theta n), \bs{S}).
  \end{align*}
\end{lemma}
\begin{proof}
  Let $\varrho = \rho^{\tensor n}$. 
  \rnote{After 0.5 seconds of thought, it wasn't immediately clear to me this is wlog, but I should probably put 30 seconds of thought into it.  Anyway, when we actually use this result we'll have reduced to projectors already, so it's cool.} Let $A_1 = A$ and $A_0 = \Id - A$. For all
  $x \in \Bin[n]$, let $A_x \in (\CC^{d \times d})^{\tensor n}$ denote
  the event defined by
  $A_x = A_{x_1} \tensor A_{x_2} \tensor \dotsb \tensor A_{x_n}$. For
  $k \in \set{0, \dotsc, n}$, let $E_k \in (\CC^{d \times d})^{\tensor n}$ be
  the event defined by
  \begin{align*}
    E_k
    &= \sum_{\substack{x \in \Bin[n] \\ \abs{x} = k}} A_x.
  \end{align*}
  Since $A$ is a projector, $A_x$ is also a projector and
  $A_x A_y = A_y A_x = 0$ for all $x, y \in \Bin[n]$ with $x \not=
  y$. Thus, each $E_k$ is a sum of orthogonal projectors, so $E_k$ is a
  projector as well and $E_k E_\ell = E_\ell E_k = 0$ for all
  $k, \ell \in \set{0, \dotsc, n}$ with $k \not= \ell$. Moreover,
  \begin{align*}
    \sum_{k=0}^n E_k
    &= \sum_{x \in \Bin[n]} A_x
      = \Id.
  \end{align*}
  Let $B \in (\CC^{d \times d})^{\tensor n}$ denote the quantum event
  defined by
  \begin{align*}
    B
    &= \sum_{k=0}^n \Pr[\bs{X} + k > \theta n] \cdot E_k.
  \end{align*}
  The statistics of the measurement $\set{E_k \mid k = 0, \dotsc, n}$
  applied to $\varrho$ follow a binomial distribution $\Binomial(n, \tr(\rho A))$, 
  so $\E_\varrho [E_k] = \Pr[\bs{S} = k]$. Hence,
  \begin{align*}
    \E_\varrho [B]
    &= \sum_{k=0}^n \Pr[\bs{X} + k > \theta n] \cdot \E_\varrho [E_k]
      = \sum_{k=0}^n \Pr[\bs{X} + k > \theta n] \cdot \Pr[\bs{S} = k]
      = \Pr[\bs{S} + \bs{X} > \theta n].
  \end{align*}
  For all $\ell \in \set{0, \dotsc, n}$,
  \begin{align*}
    \sqrt{\Id - B} \cdot E_\ell
    &= E_\ell \cdot \sqrt{\Id - B}
      = \sqrt{\Pr[\bs{X} + \ell \le \theta n]} \cdot E_\ell.
  \end{align*}
  Hence,
  \begin{align*}
    \tr \paren*{\restrict{\varrho}_{\sqrt{\Id - B}} \cdot E_\ell}
    &= \frac{1}{\E_\varrho [\ol{B}]} \cdot \tr(\sqrt{\Id - B} \cdot
      \varrho \cdot \sqrt{\Id - B} \cdot E_\ell) \\
    &= \frac{1}{\E_\varrho [\ol{B}]} \cdot \tr(E_\ell \cdot \sqrt{\Id - B} \cdot
      \varrho \cdot \sqrt{\Id - B} \cdot E_\ell) \\
    &= \frac{\Pr[\bs{X} + \ell \le \theta n]}{\E_\varrho [\ol{B}]} \cdot \tr(E_\ell \cdot \varrho \cdot E_\ell) \\
    &= \frac{\Pr[\bs{X} + \ell \le \theta n]}{\E_\varrho [\ol{B}]} \cdot
      \E_\varrho [E_\ell] \\
    &= \frac{\Pr[\bs{X} + \ell \le \theta n]}{\Pr[\bs{S} + \bs{X} \le
      \theta n]} \cdot
      \Pr[\bs{S} = \ell].
  \end{align*}
  Thus, the measurement $\set{E_k \mid k = 0, \dotsc, n}$ applied to
  $\restrict{\varrho}_{\sqrt{\Id - B}}$ yields statistics distributed
  as $(\bs{S} \mid \ol{B})$. Therefore,
  by~\Cref{prop:fidelity-bhattacharyya},
  \begin{align*}
    \Fid \paren*{\varrho, \restrict{\varrho}_{\sqrt{\Id - B}}}
    &= \sum_{k=0}^n \sqrt{\tr(\varrho \cdot E_k)} \sqrt{\tr
      \paren*{\restrict{\varrho}_{\sqrt{\Id - B}} \cdot E_k}}
      = \BC((\bs{S} \mid \bs{S} + \bs{X} \le \theta n), \bs{S}). \qedhere
  \end{align*}
\end{proof}

Using~\Cref{lem:meas-claim41}, we obtain the following ``quantum
version'' of~\Cref{thm:claim41}:

\begin{corollary}\label{cor:meas-claim41}
  Let $\rho \in \CC^{d \times d}$ represent an unknown quantum state and
  let $A \in \CC^{d \times d}$ be a projector. Let $n \in \NN$, let
  $\lambda > 0$, and let $\theta \in [0, 1]$ be an arbitrary
  threshold. Fix $p = \E_\rho [A]$ and let $\bs{S}$ and $\bs{X}$ be
  defined as in~\Cref{thm:claim41}. If $p$, $\lambda$, $n$, and $\theta$
  satisfy the conditions of~\Cref{thm:claim41}, then there exists a
  quantum event $B \in (\CC^{d \times d})^{\tensor n}$ such that
  $\E_{\rho^{\tensor n}}[B] = \Pr[\bs{S} + \bs{X} > \theta n]$ and
  \begin{align*}
    \dist[Bures] \paren*{\rho^{\tensor n}, \restrict{\rho^{\tensor
    n}}_{\sqrt{\Id - B}}}
    &\lesssim \E_{\rho^{\tensor n}}[B] \cdot \frac{\stddev[\bs
      S]}{\E[\bs X]}.
  \end{align*}
  Moreover,
  \begin{align*}
    \E_{\rho^{\tensor n}} [B]
    &\le \exp(- n \lambda (\theta - (e - 1) \E_\rho[A])).
  \end{align*}
\end{corollary}
\begin{proof}
  Let $\varrho = \rho^{\tensor n}$. By~\Cref{lem:meas-claim41}, there
  exists a quantum event $B \in (\CC^{d \times d})^{\tensor n}$ such
  that $\E_\varrho [B] = \Pr[\bs{S} + \bs{X} > \theta n]$ and
  $\Fid \paren*{\varrho, \restrict{\varrho}_{\sqrt{\Id - B}}} =
  \BC((\bs{S} \mid \bs{S} + \bs{X} \le \theta n), \bs{S})$. Note that,
  for all distributions~$\mu$ and~$\nu$,
  $1 - \BC(\mu, \nu) \le \dist[\chi^2](\mu, \nu)$. Hence,
  by~\Cref{lem:meas-claim41} and~\Cref{cor:claim41}, it follows that
  \begin{align*}
    \dist[Bures] \paren*{\rho^{\tensor n}, \restrict{\rho^{\tensor
    n}}_{\sqrt{\Id - B}}}
    &= \sqrt{2\paren*{1 -  \Fid \paren*{\varrho, \restrict{\varrho}_{\sqrt{\Id
      - B}}}}} \\
    &= \sqrt{2(1 - \BC((\bs{S} \mid \bs{S} + \bs{X} \le \theta n), \bs{S}))} \\
    &= \dist[H]((\bs{S} \mid \bs{S} + \bs{X} \le \theta n), \bs{S}) \\
    &\le \sqrt{\dist[\chi^2]((\bs{S}
      \mid \bs{S} + \bs{X} \le \theta n), \bs{S})} \\
    &\lesssim \E_{\rho^{\tensor n}}[B] \cdot \frac{\stddev[\bs
      S]}{\E[\bs X]} 
  \end{align*}
  Since $\E_\varrho [B] = \Pr[\bs{S} + \bs{X} > \theta n]$,
  \begin{align*}
    \E_\varrho [B]
    &= \Pr[\bs{S} + \bs{X} > \theta n] \\
    &\le \E[\exp(- \lambda(\theta n - \bs{S}))]
    && \text{(by $\Pr[\bs{X} > t] \le \exp(- \lambda t)$)} \\
    &= \exp(- \lambda \theta n) \E[\exp(\lambda \bs{S})] \\
    &= \exp(- \lambda \theta n) (1 - p + p e^\lambda)^n
    && \text{($\E[\exp(\lambda \bs{S})]$ is the  m.g.f.\ of $\bs{S}$)} \\
    &= \exp(- \lambda \theta n) (1 + p (e^\lambda - 1))^n \\
    &\le \exp(- \lambda \theta n) (1 + p (e - 1)\lambda)^n
    && \text{(by $e^x \le 1 + (e - 1)x$ for $x \in [0, 1]$)} \\
    &\le \exp(- \lambda \theta n) \exp((e - 1) n \lambda p)
    && \text{(by $1 + x \le e^x$ for $x \in \RR$)} \\
    &= \exp(- n \lambda (\theta - (e - 1) p)).
    && \qedhere
  \end{align*}
\end{proof}

\section{\BajanSearch} \label{sec:gentle-search}

In this section, we prove~\Cref{thm:our-gentle-search}.

\subsection{Preliminary reductions} \label{sec:reductions}
We begin with several reductions that allow us to reduce to the case of projectors, and to the case when $\eps$, $\delta$, and the $\theta_i$'s are all fixed constants.

\paragraph{Reduction to projectors.}
Let $\rho \in \CC^{d \times d}$ denote the unknown quantum state and let
$A_1, \dotsc, A_m$ be the observables in the
quantum \BajanSearch problem (which we assume are given in an online fashion). 
If we
extend the unknown state $\rho$ to $\rho \tensor \ketbra{0}{0}$, then
by Naimark's \Cref{thm:naimark-extension}, there exists a projector
$\Pi_i \in \CC^{d \times d} \tensor \CC^{2 \times 2}$ for each $A_i$
such that $\E_{\rho \tensor \ketbra{0}{0}} [\Pi_i] = \E_\rho [A_i]$ for
all $i = 1, \dotsc, m$. Since the state $\rho \tensor \ketbra{0}{0}$ can
be prepared without knowing $\rho$ and this extension increases the
dimension of the quantum system only by a constant factor, by replacing
$\rho$ by $\rho \tensor \ketbra{0}{0}$ and each $A_i$ by the
corresponding $\Pi_i$, it follows that we can assume, without loss of
generality, that the observables $A_1, \dotsc, A_m$ are projectors.

\paragraph{Reduction to $3/4$ vs.~$1/4$.}
Let $0 < \eps < \half$ be given, and recall that in the \BajanSearch problem the algorithm is presented with a stream of projector/threshold pairs $(A_i, \theta_i)$, with the goal of distinguishing the cases $\E_{\rho}[A_i] > \theta_i$ and $\E_{\rho}[A_i] \leq \theta_i - \eps$.
We may have the algorithm use \Cref{lem:amplification} (the latter part, with $\tau = 0$, $c = \theta_i - \eps/2$, $\delta = 1/4$, and $\eps$ replaced by $\eps/2$), which establishes that for some $n_0 = O(1/\eps^2)$, each $A_i$ may be replaced with a projector $B_i \in (\CC^{d \times d})^{\tensor n_0}$
satisfying
\begin{enumerate}[label=\roman*.]
\item if $\E_\rho [A_i] > \theta_i$, then $\E_{\rho^{\tensor n_0}} [B_i] > 3/4$;
\item if $\E_\rho [A_i] \le \theta_i - \eps$, then $\E_{\rho^{\tensor n_0}} [B_i] \le 1/4$.
\end{enumerate}
Thus we can reduce to the ``$3/4$ vs.~$1/4$'' version of \BajanSearch at the expense of paying an extra factor of $n_0 = O(1/\eps^2)$ in the copy complexity.
Note that the parameter~$d$ has increased to~$d^{n_0}$, as well, but (crucially) our \Cref{thm:our-gentle-search} has no dependence on the dimension parameter.\rnote{This is a point worth making, right?}

\paragraph{Reduction to a promise-problem version, with fixed~$\delta$.}
So far we have reduced proving \Cref{thm:our-gentle-search} to proving the following:
\begin{theorem} \label{thm:bajansearch2}
  There is an algorithm that, given $m \in \NN$ and $0 < \delta < \half$, first obtains $n^* = O(\log^2 m + \log(1/\delta)) \cdot \log(1/\delta)$ copies $\rho^{\otimes n^*}$ of an unknown state $\rho \in \CC^{d \times d}$.
  Next, a sequence of projectors $A_1, \dots A_m \in \CC^{d \times d}$ is presented to the algorithm (possibly adaptively).
  After each~$A_t$, the algorithm may either \emph{select}~$t$, meaning halt and output the claim ``\,$\E_{\rho}[A_t] > 1/4$'', or else \emph{pass} to the next projector.
  If the algorithm passes on all~$m$ projectors, the algorithm must claim ``\,$\E_{\rho}[A_i] \leq 3/4$ for all~$i$''.
  Except with probability at most~$\delta$, the algorithm's output is correct.
\end{theorem}
The main work we will do is to show the following similar result:
\begin{lemma} \label{lem:bajansearch3}
  There is an algorithm that, given $m  \in \NN$, first obtains $n = O(\log^2 m)$ copies $\rho^{\otimes n}$ of an unknown state $\rho \in \CC^{d \times d}$.
  Next, a sequence of projectors $A_1, \dots A_m \in \CC^{d \times d}$, obeying the promise that $\E_{\rho}[A_j] > 3/4$ for at least one~$j$, is presented to the algorithm.
  After each~$A_t$, the algorithm may either halt and \emph{select}~$t$, or else \emph{pass} to the next projector.
  With probability at least~$0.01$, the algorithm selects a~$t$ with $\E_{\rho}[A_t] \geq 1/3$.
\end{lemma}

One needs a slight bit of care to reduce  \Cref{thm:bajansearch2} to \Cref{lem:bajansearch3} while maintaining the online nature of the algorithm:
\begin{proof}[Proof of Theorem \ref{thm:bajansearch2}, assuming Lemma \ref{lem:bajansearch3}.]
    The algorithm in \Cref{lem:bajansearch3} will be used as a kind of ``subroutine'' for the main theorem.
    Our first step is to augment this subroutine in the following way:
    \begin{itemize}
            \item Given parameter~$\delta$ for the main theorem, the subroutine will use a parameter $\delta' = \delta/(C\log(1/\delta))$, where $C$ is a universal constant to be chosen later.
            \item $n$ is increased from $O(\log^2 m)$ to $n' = O(\log^2 m) + O(\log(1/\delta'))$, where the first $O(\log^2 m)$ copies of~$\rho$ are used as usual, and the additional $O(\log(1/\delta'))$ copies are reserved as a ``holdout''.
            \item If ever the subroutine is about to halt and select~$t$, it first performs a ``failsafe'' check:  It applies \Cref{lem:amplification} with $\tau = 0$, $c = .3$, $\eps = .03$, $\delta = \delta'$, and measures with the holdout copies.
                (Note that $c + \eps < 1/3$ and also $c - \eps > 1/4$.)
                If event ``$B$'' as defined in \Cref{lem:meas-claim41} occurs, the subroutine goes ahead and selects~$t$; otherwise, the algorithm not only passes, but it ``aborts'', meaning that it automatically passes on all subsequent~$A_i$'s without considering them.
    \end{itemize}
    We make two observations about this augmented subroutine:
    \begin{itemize}
        \item When run under the promise that $\E_{\rho}[A_j] > 3/4$ for at least one~$j$, it still selects a~$t$ satisfying $\E_{\rho}[A_t] \geq 1/3$ with probability at least~$0.005$.
            This is because the ``failsafe'' causes an erroneous change of mind with probability at most~$\delta'$, and we may assume $\delta' \leq 0.005$ (taking~$C$ large enough).
        \item When run \emph{without} the promise that $\E_{\rho}[A_j] > 3/4$ for at least one~$j$, the failsafe implies that the probability the algorithm ever selects a~$t$ with $\E_{\rho}[A_t] < 1/4$ is at most~$\delta'$.
    \end{itemize}
    With the augmented subroutine in hand, we can now give the algorithm
    that achieves \Cref{thm:bajansearch2}.  The algorithm will obtain
    $n^* = n' \cdot L$ copies of~$\rho$, where $L = O(\log(1/\delta))$;
    these are thought of as $L$ ``batches'', each with of $n'$~copies.
    As the projectors~$A_i$ are presented to the algorithm, it will run
    the augmented subroutine ``in parallel'' on each batch.  If any
    batch wants to halt accept a certain~$A_t$, then the overall
    algorithm halts and outputs ``\,$\E_\rho[A_t] > 1/4$''.  Otherwise,
    if all the batches pass on~$A_t$, so too does the overall algorithm.
    Of course, if the overall algorithm passes on all~$A_i$'s, it
    outputs ``\,$\E_\rho[A_i] \leq 3/4$ for all~$i$''.

    We now verify the correctness of this algorithm.  First, \emph{if}
    there exists some $A_j$ with $\E_\rho[A_j] > 3/4$, the probability
    of the algorithm wrongly outputting ``\,$\E_\rho[A_i] \leq 3/4$ for
    all~$i$'' is at most $(1-.005)^L$, which can be made smaller
    than~$\delta$ by taking the hidden constant in
    $L = O(\log(1/\delta))$ suitably large.  On the other hand, thanks
    to the ``failsafe'' and a union bound, the probability the algorithm
    ever wrongly outputs ``\,$\E_\rho[A_t] > 1/4$'' is at
    most~$L \delta' = L \cdot \delta/(C\log(1/\delta))$, which is again
    at most~$\delta$ provided~$C$ is taken large enough.
\end{proof}

\subsection{The main algorithm (proof of Lemma \ref{lem:bajansearch3})} 

In this section, we will prove \Cref{lem:bajansearch3}. Let $n = n(m)$
and $\lambda = \lambda(m)$ be parameters to be fixed later and let
$\theta = 2/3$.
As stated in \Cref{lem:bajansearch3}, we may explicitly assume there exists
$i \in [m]$ with $\E_\rho[A_i] \ge 3/4$.  For each projector $A_i$, let
$B_i$ denote the event obtained from \Cref{lem:meas-claim41}. The
algorithm proceeds as follows:

\begin{framed}
  Let $\varrho$ denote the current quantum state, with $\varrho = \rho^{\tensor n}$ initially.
  Given projector~$A_i$, let $B_i$ be the event obtained from \Cref{lem:meas-claim41}.
  Measure the current state $\varrho$ with $(\ol{B}_i, B_i)$ using the canonical implementation.
  If~$B_i$ occurs, halt and select~$i$; otherwise, pass.
\end{framed}
\noindent Note that the $n$ copies of $\rho$ are only prepared once and reused,
and that the current state~$\varrho$ collapses to a new state after each measurement.

The algorithm has the following modes of failure:
\begin{enumerate}[label=(\alph*)]
\item[(FN)] the algorithm passes on every observable because the event
  $\ol{B}_i$ occurs for every $i \in [m]$;
\item[(FP)] the algorithm picks an observable $A_j$ with
  $\E_\rho [A_j] < 1/3$.
\end{enumerate}

We want to show that the algorithm does not make errors of type FP or FN
with probability at least $0.1$. To this end, we introduce the following
notation.
\begin{notation}
  For $i = 1, \dotsc, m$, let:
  \begin{enumerate}
  \item $\bs{S}_i$ be a random variable distributed as
    $\Binomial(n, \E_\rho [A_i])$;
  \item $p_i = \E_{\rho^{\tensor n}} [B_i]$ be the probability that
    $B_i$ would occur if $\rho^{\tensor n}$ were measured with
    $(\ol{B}_i, B_i)$;
  \item $\varrho_0 = \rho^{\tensor n}$ and let $\varrho_i$ be the
    quantum state after the $i$th measurement, 
    \emph{conditioned} on the event $\ol{B}_j$ occurring for all $1 \le j \le i$;
  \item $r_i = \E_{\varrho_{i-1}}[\ol{B}_i]$ be the probability that the
    event $\ol{B}_i$ occurs assuming all the events $\ol{B}_j$ with $1 \le j
    \le i - 1$ occurred;
  \item $q_i = r_1 \dotsm r_i$ be the probability that \emph{all} of the
    events $\ol{B}_j$ with $1 \le j \le i$ occur;
  \item $s_i = q_{i-1} \cdot \E_{\varrho_{i-1}}[B_i]$ be the probability
    of observing outcomes $\ol{B}_1, \dotsc, \ol{B}_{i-1}, B_i$.
  \end{enumerate}
\end{notation}

Note that the $p_i$'s refer to a ``hypothetical,'' whereas the $r_i$'s,
$q_i$'s, and $s_i$'s concern what actually happens over the course of the
algorithm.  In particular, $q_m$ is the probability that the algorithm
passes on every observable
. 
The following claim shows that, as long as the noise expectation
$\E[\bs{X}] = 1/\lambda$ used in \Cref{lem:meas-claim41} is sufficiently
large, the probability of a false negative (FN) is bounded above
by~$4/5$:
\begin{claim}\label{claim:there-is-t}
  For $\E[\bs X] = \Omega(\sqrt{n})$, there exists $t \in [m]$ such that
  $q_t \le 4/5$. Moreover, if $t > 1$, then $q_{t-1} \ge 3/4$ and
  $p_1 + \dotsb + p_{t-1} \le 1/4$.
\end{claim}
\begin{proof}
  By~\Cref{lem:meas-claim41},
  $p_i = \E_{\rho^{\tensor n}}[B_i] = \Pr[\bs{S}_i + \bs{X} > \theta
  n]$. Let $k \in [m]$ be such that $\E_\rho [A_k] \ge 3/4$. Thus,
  $\bs{S}_k$ is a binomial random variable with mean at least
  $3/4$. Since $\theta = 2/3 < 3/4$, if $n$ is taken to be a
  sufficiently large constant,
\begin{align*}
  p_k
  &= \Pr[\bs{S}_k + \bs{X} > \theta n]
    \ge \Pr[\bs{S}_k > (2/3) n]
    \ge 1 - \exp(-1/4).
\end{align*}
Therefore, there exists a minimal $t \in [m]$ such that
$(1 - p_1) \dotsm (1 - p_t) \le \exp(-1/4)$. If $t = 1$, then
$q_1 = 1 - p_1 \le \exp(-1/4) \le 4/5$. Otherwise, since $t$ is minimal,
it follows that $(1 - p_1) \dotsm (1 - p_{t-1}) \ge \exp(-1/4)$. Hence,
\begin{align*}
  \exp(-1/4)
  &\le (1 - p_1) \dotsm (1 - p_{t-1})
    \le \exp(- (p_1 + \dotsb + p_{t-1})),
\end{align*}
whence $p_1 + \dotsb + p_{t-1} \le 1/4$. Thus, by~\Cref{lem:damage}
and~\Cref{cor:meas-claim41},
\begin{align*}
  \abs{(1 - p_1) \dotsm (1 - p_t) - q_t}
  &\le 2 \sum_{i=1}^{t-1} \dist[\tr] \paren*{\rho^{\tensor n}, \restrict{\rho^{\tensor
    n}}_{\sqrt{\Id - B_i}}} \\
  &\lesssim \sum_{i=1}^{t-1} \E_{\rho^{\tensor n}}[B_i] \cdot \frac{\stddev[\bs
    S_i]}{\E[\bs X]}
  \le \frac{\sqrt{n}}{\E[\bs X]} \cdot (p_1 + \dotsc + p_{t-1})
  \le \frac{1}{4} \cdot \frac{\sqrt{n}}{\E[\bs X]}.
\end{align*}
By a similar argument,
\begin{align*}
  \abs{(1 - p_1) \dotsm (1 - p_{t-1}) - q_{t-1}}
  &\lesssim \frac{\sqrt{n}}{\E[\bs X]} \cdot (p_1 + \dotsc + p_{t-2})
    \le \frac{1}{4} \cdot \frac{\sqrt{n}}{\E[\bs X]}.
\end{align*}
Therefore, since $3/4 < \exp(-1/4) < 4/5$, we have $q_t \le 4/5$ and
$q_{t-1} \ge 3/4$, for $\E[\bs X] = \Omega(\sqrt{n})$.
\end{proof}

Assuming $\E[\bs X] = \Omega(\sqrt{n})$, let $t \in [m]$ be as
in~\Cref{claim:there-is-t}. Since $q_m \le q_t \le 4/5$, it follows that
the probability the algorithm makes an FN error is at most $4/5$. In
fact, since $q_t \le 4/5$, the algorithm will pick an index $i \le t$
with probability at least $1/5$. Thus, to show that the algorithm
succeeds w.p.\ at least $0.1$, it suffices to show that w.h.p.\ the
algorithm does not pick an index $i \in \cB$, where $\cB \subset [m]$ is
the subset defined by
\begin{align*}
  \cB
  &= \set{ i \in [m] \mid 1 \le i \le t\ \text{and}\ \E_\rho [A_i] < 1/3}.
\end{align*}
First, we show that an event $B_i$ with $i \in \cB$ is unlikely to occur
when the initial state $\rho^{\tensor n}$ is measured according to
$(\ol{B}_i, B_i)$:
\begin{claim}\label{claim:bad-things-dont-happen}
  Let $\eta \in (0, 1]$, to be specified later. If $n$ is of order
  $O(\log^2 (m / \eta))$, then $p_i \le (\eta/m)^2$ for all $i \in \cB$.
\end{claim}
\begin{proof}
  By~\Cref{cor:meas-claim41}, for all $i \in [m]$,
  \begin{align*}
    p_i
    &= \E_{\rho^{\tensor n}} [B_i]
      \le \exp(- n \lambda (\theta - (e - 1) \E_\rho[A_i])).
  \end{align*}
  Since $\theta = 2/3$ and $i \in \cB$, we have $\E_\rho[A_i] < 1/3$ and
  $\theta - (e - 1)\E_\rho[A_i] \ge 0.09$. Since
  $n \lambda = \Omega(\sqrt{n})$, there exists a constant $C > 0$ such
  that $n \lambda \ge C \sqrt{n}$. Thus,
  \begin{align*}
    p_i
    &= \E_{\rho^{\tensor n}} [B_i]
      \le \exp(- 0.09 C \sqrt{n}).
  \end{align*}
  Therefore, if $n \ge \log^2((m/\eta)^2) / (0.09 C)^2$, then
  $p_i \le (\eta / m)^2$.
\end{proof}

Next, we show that the algorithm picks an index $i \in [t]$ such that
$\E_\rho[A_i] \ge 1/3$ with probability at least $0.03$, proving
\Cref{lem:bajansearch3}.

\begin{proof}[Proof of Lemma \ref{lem:bajansearch3}]
  Fix $\eta = 0.01$, so that $n = O(\log^2 m)$ as promised. By
  \Cref{lem:fid-ineq},
  \begin{align*}
    1
    &\le \sqrt{q_t} \Fid(\rho^{\tensor n}, \varrho_t) + \sum_{i=1}^t
      \sqrt{s_i} \sqrt{p_i}.
  \end{align*}
  By \Cref{claim:bad-things-dont-happen} and the Cauchy--Schwarz
  inequality,
  \begin{align*}
    \sum_{i=1}^t \sqrt{s_i} \sqrt{p_i}
    &\le \frac{\eta}{m} \sum_{i \in \cB} \sqrt{s_i} + \sum_{i \not\in
      \cB} \sqrt{s_i} \sqrt{p_i}
      \le \eta + \sqrt{\sum_{i \not\in \cB} s_i} \sqrt{\sum_{i \not\in \cB} p_i},
  \end{align*}
  where $i \not\in \cB$ denotes $i \in [t] \setminus \cB$. By
  \Cref{claim:there-is-t}, $p_1 + \dotsb + p_t \le 1/4$. Hence,
  \begin{align*}
    1 - \sqrt{q_t} \Fid(\rho^{\tensor n}, \varrho_t) - \eta
    &\le \sqrt{\sum_{i \not\in \cB} s_i} \sqrt{\sum_{i \not\in \cB} p_i}
      \le \half \sqrt{\sum_{i \not\in \cB} s_i}. 
  \end{align*}
  Since $\Fid(\rho^{\tensor n}, \varrho_t) \le 1$, $\eta = 0.01$, and, by
  \Cref{claim:there-is-t}, $q_t \le 4/5$, it follows that
  \begin{align*}
    \half \sqrt{\sum_{i \not\in \cB} s_i}
    \ge 0.99 - \sqrt{4/5}
    &\implies \sum_{i \not\in \cB} s_i
      \ge 4 \cdot (0.99 - \sqrt{4/5})^2
      \ge 0.03.
  \end{align*}
  Since $\sum_{i \not\in \cB} s_i$ is the probability that the algorithm
  returns an index $i \in [t]$ with $\E_\rho[A_i] \ge 1/3$, it follows
  that the algorithm is correct with probability at least $0.03$.
\end{proof}

\section{Shadow Tomography and Hypothesis Selection}    \label{sec:applications}

\subsection{Shadow Tomography}
We begin by describing how to deduce our online Shadow Tomography result, \Cref{thm:shadow-tomog}, from our online \BajanSearch result, \Cref{thm:our-gentle-search}.
As mentioned earlier, this deduction is known~\cite{Aar20} to follow almost immediately from a mistake-bounded learning algorithm for quantum states due to Aaronson, Chen, Hazan, Kale, and Nayak~\cite{ACHKN19}, described below.
We will fill in a few details that are not spelled out in~\cite{Aar20}.

\paragraph{Mistake-bounded learning scenario.}
Consider the following kind of interaction between a ``student'' and a ``teacher'', given parameters $d \in \NN$ and $0 < \eps < \frac12$.
There is a quantum state $\rho \in \CC^{d \times d}$ that is unknown to the student (and possibly also unknown to the teacher).
The teacher presents a sequence of quantum events $A_1, A_2, A_3, \dots$ (possibly adaptively) to the student.
Upon receiving $A_t$, the student must output a prediction $\wh{\mu}_t$ of $\mu_t = \E_{\rho}[A_t]$.
After each prediction, the teacher must either ``pass'', or else declare a ``mistake'' and supply a value~$\mu'_t$.
\begin{theorem}[\cite{ACHKN19}]                                     \label{thm:achkn}
    Assume the following \textbf{Teacher Properties} hold for each~$t$:
    \begin{itemize}
        \item If $|\wh{\mu}_t - \mu_t| > \eps$, the teacher always declares ``mistake''.
        \item If $|\wh{\mu}_t - \mu_t| \leq \frac34 \eps$, the teacher always passes.
        \item If the teacher ever declares ``mistake'', the supplied value $\mu'_t$ always satisfies $|\mu'_t - \mu_t| \leq \frac14 \eps$.
        \item (If $\frac34 \eps < |\wh{\mu}_t - \mu_t| \leq \eps$, the
          teacher may either pass or declare a mistake; but, if the latter, recall that $|\mu'_t - \mu_t| \leq \frac14 \eps$.)
    \end{itemize}
    Then there is an algorithm for the student that causes at most $C_0 (\log d)/\eps^2$ ``mistakes'' (no matter how many events are presented), where $C_0$ is a universal constant.
\end{theorem}
\noindent The above theorem is similar to, but not quite the same, as ``Theorem~1'' in~\cite{ACHKN19}.
However it is easy to check that \cite{ACHKN19}'s Section~3.3 (``Proof of Theorem~1'') applies equally well to establish \Cref{thm:achkn} above.%
\footnote{Briefly: the RTFL/MMW algorithm will only do an update in the ``mistake'' rounds.
The loss is taken to be $|\wh{\mu}_t - \mu'_t|$.  On any mistake, we have $|\wh{\mu}_t - \mu_t| > \frac34 \eps$ and $|\mu'_t - \mu_t| \leq \frac14 \eps$, hence the student incurs loss at least $\frac12 \eps$.
On the other hand, answering according to the true $\mu_t$ would only incur loss at most $\frac14\eps$.
The regret calculation bounding the number of mistakes is now the same.}

To use this theorem for the online Shadow Tomography problem, it only remains for the Shadow Tomography algorithm \emph{to implement the teacher's role itself}, given copies of~$\rho$.
This will be done using our \BajanSearch algorithm; let us first slightly upgrade it so that (i)~it is concerned with $\E_\rho[A_i] \approx \theta_i$ rather than $\E_\rho[A_i] < \theta_i$; (ii)~if it finds $j$ with $\E_\rho[A_j] \not \approx \theta_j$, then it also reports a very good estimate of $\E_\rho[A_j]$.
\begin{lemma}                                       \label{lem:bajanupgrade}
    Consider the version of quantum Threshold Search where the inputs are the same, but the algorithm should correctly (except with probability at most~$\delta$) output:
    \begin{itemize}
        \item ``\,$\abs{\E_\rho[A_j] - \theta_j} > \frac34 \eps$, and in fact $\abs{\E_\rho[A_j] - \mu'_j} \leq \frac14\eps$'', for some particular $j$ and value~$\mu'_j$; or else,
        \item ``\,$\abs{\E_\rho[A_i] - \theta_i} \leq \eps$ for all~$i$''.
    \end{itemize}
    Then as in \Cref{thm:our-gentle-search}, the problem can be solved in an online fashion using
    \[
            \phantom{\qquad \text{($\textsc{l} = \log(1/\delta)$)}}
        \nts'(m,\eps,\delta) = \frac{\log^2 m + \textsc{l}}{\eps^2} \cdot O(\textsc{l}) \qquad \textnormal{($\textsc{l} = \log(1/\delta)$)}
    \]
    copies of~$\rho$.
\end{lemma}
\begin{proof}
    Given $m, \eps, \delta$, we obtain $n = \nts(2m, \frac14 \eps, \delta/2) + c\log(1/\delta)/\eps^2$ copies of~$\rho$, where $c$ is a universal constant to be specified later.
    This $n$ indeed has the asymptotic form of $\nts'$ given above.
    We save the $c\log(1/\delta)/\eps^2$ copies as a ``holdout'', and use the remaining copies to apply \Cref{thm:our-gentle-search} (with parameters $2m, \frac14 \eps, \delta/2$), converting our given observable/threshold pairs $(A_1, \theta_1), \dots, (A_m, \theta_m)$ to a ``simulated input'' of
    \[
        (A_1, \theta_1 + \eps), (\Id-A_1, 1-\theta_1 + \eps), \dots, (A_m, \theta_m + \eps), (\Id-A_m, 1-\theta_m + \eps).
    \]
    Except with probability at most~$\delta/2$ we get a correct answer from the simulation, from which we can derive a correct final output as described below.

    If the simulation passes on all $2m$ pairs, then \Cref{thm:our-gentle-search} tells us that  we must have
    \[
        \E_\rho[A_i] \leq \theta_i + \eps\qquad \text{and} \qquad
        \E_\rho[\Id-A_i] \leq 1-\theta_i + \eps
    \]
    for all~$i$, and therefore we may correctly output ``\,$\abs{\E_\rho[A_i] - \theta_i} \leq \eps$ for all~$i$''.

    On the other hand, suppose the simulation halts by outputting
    \[
        \text{``}\E_\rho[A_j] > \theta_j + \eps - \tfrac14 \eps\text{''} \qquad \text{or} \qquad  \text{``}\E_\rho[\Id - A_j] > 1 - \theta_j + \eps - \tfrac14 \eps\text{''}
    \]
    for some particular~$j$.
    Then our algorithm can correctly output ``\,$\abs{\E_\rho[A_j] - \theta_j} > \frac34 \eps$''.
    Furthermore, at this point the algorithm may use the holdout copies of~$\rho$ to obtain an estimate $\mu'_j$ of $\E_\rho[A_j]$ (in the naive way\rnote{make a proper reference to a generalized \Cref{lem:amplification}}) that satisfies $\abs{\E_\rho[A_j] - \mu'_j} \leq \frac14\eps$ except with probability at most~$\delta/2$, provided $c$ is large enough.
\end{proof}

With \Cref{lem:bajanupgrade} in place, we can obtain our online Shadow Tomography algorithm:
\begin{proof}[Proof of Theorem \ref{thm:shadow-tomog}.]
    Define
    \[
        R = \lceil C_0 (\log d)/\eps^2 \rceil + 1, \quad \delta_0 = \delta/R, \quad n_0 = \nts'(m, \eps, \delta_0).
    \]
    The number of copies of~$\rho$ used by our online Shadow Tomography algorithm will be $n = R n_0$, which is
    indeed
    \[
        n = \frac{(\log^2 m + \textsc{l})(\log d)}{\eps^4} \cdot O(\textsc{l})
    \]
    for  $\textsc{l} = \log(\tfrac{\log d}{\delta \eps})$, as claimed.

    Upon receiving $n$ copies of~$\rho$, our Shadow Tomography algorithm partitions it into~$R$ ``batches'' of size $n_0$ each.
    The idea is that each batch will be devoted to (up to) one ``mistake'' of the ``student''.
    We now describe the algorithm, and then give its analysis.

    To begin, recall that our Shadow Tomography algorithm receives
    the input quantum events $A_1, A_2, \dots$ in an online fashion.  As it
    receives them, it will run the following online algorithms
    concurrently:
    \begin{itemize}
        \item the mistake-bounded learning algorithm of \Cref{thm:achkn} (implementing the student's algorithm);
        \item the \BajanSearch algorithm from \Cref{lem:bajanupgrade} (to implement the teacher), initially using only the first batch of $\rho^{\otimes n_0}$.
        \end{itemize}

        The algorithm simulates both the teacher and student roles of
        the mistake-bounded setting of~\cite{ACHKN19} and runs in
        rounds. A new round is started whenever the teacher declares a
        mistake and a fresh batch of $n_0$ copies of the state $\rho$ is
        used by the teacher in each round. When it receives input $A_t$,
        the algorithm runs the next iteration of the mistake-bounded
        learning algorithm of \Cref{thm:achkn} to get the student's
        prediction $\wh{\mu}_t$. Then it runs the next iteration of the
        \BajanSearch algorithm from \Cref{lem:bajanupgrade} with input
        ($A_t$, $\wh{\mu}_t$); the estimates $\wh{\mu}_t$ output by the
        student serve as the~$\theta_t$ threshold values used in
        \Cref{lem:bajanupgrade}.

    Whenever the \BajanSearch algorithm ``passes'' on an $(A_t,\wh{\mu}_t)$ pair, the teacher also ``passes'', and $\wh{\mu}_t$ serves as the Shadow Tomography algorithm's final estimate for $\E_\rho[A_t]$.
    On the other hand, if the \BajanSearch algorithm outputs ``\,$\abs{\E_\rho[A_t] - \wh{\mu}_t} > \frac34 \eps$, and in fact $\abs{\E_\rho[A_t] - \mu'_t} \leq \frac14\eps$'', then the teacher will declare a ``mistake'' and supply the value~$\mu'_t$ to the student.
    This $\mu'_t$ will also serve as the Shadow Tomography algorithm's final estimate for $\E_\rho[A_t]$.
    Furthermore, at this point the teacher will abandon any remaining copies of~$\rho$ in the current batch, and will use a ``fresh'' batch~$\rho^{\otimes n_0}$ for the subsequent application of \Cref{lem:bajanupgrade}.
    We refer to this as moving on to the next ``round''.

    Let us now show that with high probability there are at most~$R-1$ mistakes and hence at most $R$ rounds.
    (If the Shadow Tomography algorithm tries to proceed to an $(R+1)$th
    round, and thereby runs out of copies of~$\rho$, we simply declare
    an overall failure.)

    The total probability of error made by the Threshold Search
    algorithm within each round is bounded by $\delta_0$. By a union
    bound, the probability of any incorrect answer over all $R$ rounds
    is at most $R \delta_0$, i.e., at most $\delta$.
    Below we will show that if there are no
    incorrect answers, then the ``Teacher Properties'' of
    \Cref{thm:achkn} hold, and therefore the total number of mistakes is
    indeed at most $\lceil C_0 (\log d)/\eps^2 \rceil = R - 1$ with
    probability at least~$1-\delta$.

    It remains to verify that --- assuming correct answers from all uses
    of \Cref{lem:bajanupgrade} --- our Shadow Tomography algorithm
    satisfies the Teacher Properties of \Cref{thm:achkn} and also that
    all~$m$ estimates for $\E_\rho[A_i]$ produced by the algorithm are
    correct to within an additive error~$\eps$. Let us first note that
    within each round of the Shadow Tomography algorithm, we never
    supply more than $m$ quantum events to the Threshold Search
    algorithm from \Cref{lem:bajanupgrade}. The main point to observe is
    that if our \BajanSearch routine from \Cref{lem:bajanupgrade} ever
    passes on some $(A_t, \wh{\mu}_t)$ pair, it \emph{must} be that
    $\abs{\E_\rho[A_t] - \wh{\mu}_t} \leq \eps$; the reason is that
    passing implies the \BajanSearch algorithm is prepared to output
    ``\,$\abs{\E_\rho[A_i] - \theta_i} \leq \eps$ for all~$i$''.  On the
    other hand, it's immediate from \Cref{lem:bajanupgrade} that if the
    teacher declares ``mistake'' on some~$(A_t,\wh{\mu}_t)$ pair, then
    indeed we have $\abs{\E_\rho[A_t] - \wh{\mu}_t} > \frac34 \eps$, and
    the supplied correction~$\mu'_t$ satisfies
    $\abs{\E_\rho[A_t] - \mu'_t} \leq \frac14\eps$ (as is necessary for
    the Teacher Properties, and is more than sufficient for the Shadow
    Tomography guarantee).
\end{proof}


\subsection{Hypothesis Selection}

In this section we establish our quantum Hypothesis Selection result, \Cref{thm:hypothesis-selection}.
This theorem effectively has three different bounds, and we prove them via \Cref{prop:hypoth1,prop:hypoth2,prop:hypoth3}.

Recall that in the Hypothesis Selection problem there are given fixed hypothesis states $\sigma_1, \dots, \sigma_m \in \CC^{d \times d}$, as well as access to copies of an unknown state $\rho \in \CC^{d \times d}$.
We write
\[
    \eta = \min_{i} \{\dist[tr](\rho, \sigma_i)\}, \qquad i^* = \argmin_i \{\dist[tr](\rho, \sigma_i)\},
\]
with the quantity $\eta$ being unknown to the algorithm.
Recall that the Holevo--Helstrom theorem implies that for each pair $i \neq j$, there is a quantum event $A_{ij}$ 
such that
\[
    \E_{\sigma_i}[A_{ij}] - \E_{\sigma_j}[A_{ij}] = \dist[tr](\sigma_i, \sigma_j),
\]
and furthermore we may take $A_{ji} = \ol{A}_{ij} = \Id - A_{ij}$.
These events \emph{are} known to the algorithm.

One way to solve the quantum Hypothesis Selection problem is to simply use Shadow Tomography as a black box.
Given parameters $0 < \eps, \delta < \frac12$ for the former problem, we can run Shadow Tomography with parameters $\eps/2$, $\delta$, and the $\binom{m}{2}$ quantum events $(A_{ij} : i < j)$.
Then except with probability at most~$\delta$, we obtain values $\wh{\mu}_{ij}$ with $\abs{\E_\rho[A_{ij}] - \wh{\mu}_{ij}} \leq \eps/2$ for all $i,j$.
Now we can essentially use any classical Hypothesis Selection algorithm; e.g., the ``minimum distance estimate'' method of Yatracos~\cite{Yat85}.
We select as our hypothesis $\sigma_k$, where $k = \argmin_\ell
\wh{\Delta}_\ell$ is a minimizer of
\[
    \widehat{\Delta}_\ell = \max_{i < j}\ \abs{\E_{\sigma_\ell}[A_{ij}] - \wh{\mu}_{ij}}.
\]
Recalling $\eta = \dist[tr](\rho, \sigma_{i^*})$, we have
\begin{equation}    \label[ineq]{ineq:Delta}
    \widehat{\Delta}_k \leq \widehat{\Delta}_{i^*} \leq \max_{i < j}\set{\abs{\E_{\sigma_{i^*}}[A_{ij}] - \E_{\rho}[A_{ij}]} + \abs{\E_{\rho}[A_{ij}] - \wh{\mu}_{ij}}} \leq \eta + \eps/2,
\end{equation}
where the last inequality used the Holevo--Helstrom theorem again, and the Shadow Tomography guarantee.
We now obtain the following result (with the proof being an almost verbatim repeat of the one in \cite[Thm.~6.3]{DL01}):
\begin{proposition} \label{prop:hypoth1}
    The above-described method selects $\sigma_k$ with
    $
        \dist[tr](\rho, \sigma_k) \leq 3\eta + \eps
    $
    (except with probability at most~$\delta$),
    using a number of copies of~$\rho$ that is the same as in Shadow Tomography (up to constant factors).
\end{proposition}
\begin{proof}
    By the triangle inequality for $\dist[tr]$ we have
    \begin{align*}
        \dist[tr](\sigma_k, \rho) \leq \eta + \dist[tr](\sigma_{k}, \sigma_{i^*}) 
        &= \eta + \abs{\E_{\sigma_{k}}[A_{ki^*}] - \E_{\sigma_{i^*}}[A_{ki^*}]} \\
        &\leq\eta + \abs{\E_{\sigma_{k}}[A_{ki^*}] - \wh{\mu}_{ki^*}} + \abs{\E_{\sigma_{i^*}}[A_{ki^*}] - \wh{\mu}_{ki^*}} \leq \eta + \widehat{\Delta}_k + \widehat{\Delta}_{i^*},
    \end{align*}
    and the result now follows from \Cref{ineq:Delta}.
\end{proof}

Now we give a different, incomparable method for Hypothesis Selection.
It will use the following ``decision version'' of quantum Threshold Search, which we prove at the end of \Cref{app:threshold-decision} (see \Cref{cor:quantum-or}):
\begin{corollary}                                       \label{cor:cor}
    Consider the scenario of quantum \BajanSearch (i.e., one is given parameters $0 < \eps_0, \delta_0 < \frac12$, and $m_0$~event/threshold pairs $(A_i,\theta_i)$).
    Suppose one is further given values $\eta_1, \dots, \eta_{m_0}$.
    Then using just $n_0 = O(\log(m_0/\delta_0)/\eps_0^2)$ copies of~$\rho$, one can correctly output (except with probability at most~$\delta$):
    \begin{itemize}
        \item ``there exists~$j$ with $\abs{\E_\rho[A_j] - \theta_j} > \eta_j $''; or else,
     \item ``\,$\abs{\E_\rho[A_i]  - \theta_i} \leq \eta_i + \eps$ for all~$i$''.
    \end{itemize}
    Indeed, the algorithm can be implemented by a projector applied to $\rho^{\otimes n_0}$.
\end{corollary}
Returning to Hypothesis Selection, let us define
\[
    \Delta_k = \max_{i < j}\ \abs{\E_{\sigma_k}[A_{ij}] - \E_{\rho}[A_{ij}]},
\]
and note that $\Delta_{i^*} \leq \eta$, by the Holevo--Helstrom theorem.  Let us
also assume the algorithm has a candidate upper bound $\ol{\eta}$
on~$\eta$.  Now suppose our algorithm is able to find~$\ell$ with
$\Delta_\ell \leq \ol{\eta} + \eps$.  Then the proof of
\Cref{prop:hypoth1} similarly implies that $\sigma_\ell$ satisfies
$\dist[tr](\sigma_\ell, \rho) \leq 2\eta + \ol{\eta} + \eps$.

Now let $\cT_k$ denote the following Threshold Decision instance (as in \Cref{cor:cor}): \mbox{$\eps_0 = \eps$}, $\delta_0 = 1/3$, $m_0 = \binom{m}{2}$, the quantum events are all the~$A_{ij}$'s, the thresholds are $\theta_{ij} = \E_{\sigma_k}[A_{ij}]$, each ``$\eta_{ij}$'' is~$\ol{\eta}$.
Then \Cref{cor:cor} gives us a projector~$B_k$ on $(\CC^{d})^{\otimes n_0}$, where $n_0 = O(\log(m)/\eps^2)$, with the following property:
When it is used to measure~$\varrho = \rho^{\otimes n_0}$,
\begin{equation}    \label[ineq]{ineq:imps}
    \Delta_k \leq \ol{\eta}   \implies  \E_{\varrho}[\ol{B}_k] \geq 2/3, \qquad
     \E_{\varrho}[\ol{B}_k]  > 1/3  \implies \Delta_k 
     \leq \ol{\eta} + \eps.
\end{equation}
We can now apply our Threshold Search routine to the $\ol{B}_k$'s (with
all thresholds $\theta_k = 1/2$), using $\nts(m, 1/6, \delta')$ copies
of~$\varrho$, for some $\delta' \in (0, 1]$ to be specified shortly.
Provided that indeed $\eta \leq \ol{\eta}$, we know there is at least one~$k$ (namely $k = i^*$) with $\Delta_k \leq \ol{\eta}$; thus except with probability at most~$\delta'$, the Threshold Search routine will find an~$\ell$ with $\Delta_\ell \leq \ol{\eta} + \eps$.

If we wish to assume our Hypothesis Selection algorithm ``knows''~$\eta$, then we are done.
Otherwise, we can search for the approximate value of~$\eta$, as follows:
We perform the above routine with $\ol{\eta} = 1, \frac12, \frac14,
\frac18, \dots$, using fresh copies for each iteration and stopping either when Threshold Search fails to find any~$\ell$ or when $\ol{\eta} \leq \eps$.
If we stop for the first reason, we know that our second-to-last~$\ol{\eta}$ is at most~$2\eta$; if we stop for the second reason, we know that~$\eta \leq \eps$.
Either way, assuming no failure on any of the Threshold Searches, we end with a guarantee of $\dist[tr](\sigma_\ell, \rho) \leq 4\eta + 2\eps$.
To bound the overall failure probability we take $\delta' = \delta/\Theta(\log(1/{\max\{\eta, \eps\}}))$.
It's easy to check that the geometric decrease of $\ol{\eta}$ means we
only use $O(\nts(m, 1/6, \delta') \cdot \log(1/{\max\set{\eta, \eps}}))$
copies of~$\varrho$, which is $O(\nts(m, 1/6, \delta') \cdot \log(1/{\max\set{\eta, \eps}})) n_0$ copies of~$\rho$.
Finally, by tuning the constants we can make the final guarantee $\dist[tr](\sigma_\ell, \rho) \leq 3.01\eta + \eps$.
We conclude:
\begin{proposition} \label{prop:hypoth2}
    The above-described method selects $\sigma_\ell$ with
    $
        \dist[tr](\rho,\sigma_\ell) \leq 3.01\eta + \eps
    $
    (except with probability at most~$\delta$), using
    \[
        n = \frac{\log^3 m + \log(\textsc{L}/\delta) \cdot \log m}{\eps^2} \cdot
        O(\textsc{l} \cdot \log(\textsc{l}/\delta))
    \]
    copies of~$\rho$, where $\textsc{l} = \log(1/{\max\set{\eta,\eps}})$.
\end{proposition}

It remains to establish the last part of \Cref{thm:hypothesis-selection}, which operates under the assumption that
$\eta < \frac12 (\alpha - \eps)$, where $\alpha = \min_{i \neq j} \dist[tr](\sigma_i,\sigma_j)$.
Writing $\ol{\eta} = \frac12 (\alpha - \eps)$ (which is a quantity known to the algorithm), we have  $\Delta_{i^*} \le \eta \leq \ol{\eta}$, but $\Delta_k > \ol{\eta} + \eps$ for all $k\neq i^*$; the reason for this last claim is that
\[
    \Delta_k \geq \abs{\E_{\sigma_k}[A_{i^*k}] - \E_{\rho}[A_{i^*k}]} \geq \abs{\E_{\sigma_k}[A_{i^*k}] - \E_{\sigma_{i^*}}[A_{i^*k}]} - \eta = \dist[tr](\sigma_{i^*}, \sigma_k) - \eta \geq \alpha - \eta = 2\ol{\eta} +\eps - \eta > \ol{\eta} + \eps
\]
where the second inequality above used the Holevo--Helstrom theorem and
$\eta = \dist[tr](\rho,\sigma_{i^*})$ and the last inequality used
$\ol{\eta} > \eta$.  Now if we perform Threshold Search to achieve
\Cref{ineq:imps} as before, except that we select
$\delta_0 = \delta/(4m)$ rather than~$1/3$, we'll get projectors
$B_1, \dots, B_m$ on $(\CC^d)^{n}$ for $n = O(\log(m/\delta)/\eps^2)$
such that, for $\varrho = \rho^{\otimes n}$,
\[
    \E_{\varrho}[\ol{B}_{i^*}] \geq 1-\delta/(4m), \qquad
     \E_{\varrho}[\ol{B}_k]  \leq \delta/(4m) \quad \forall k \neq i^*.
\]
It remains to apply the Quantum Union
Bound (specifically, \Cref{cor:quantum-union-bound}) to $B_1, \dots, B_m$ and $\varrho$
to pick out~$i^*$ except with probability at most
$4 \sum_i \delta/(4m) \leq \delta$.  We conclude:
\begin{proposition} \label{prop:hypoth3}
    Using the assumption $\eta < \frac12 (\alpha - \eps)$, where $\alpha = \min_{i \neq j} \dist[tr](\sigma_i,\sigma_j)$, the above-described method selects $\sigma_{i^*}$
    (except with probability at most~$\delta$), using
    $
        n = O(\log(m/\delta)/\eps^2)
    $
    copies of~$\rho$.
\end{proposition}

\section*{Acknowledgments}
The authors are very grateful to John Wright, whose early contributions
to this work provided invaluable understanding of the problem. We also
thank the anonymous reviewers for their careful reading of the
manuscript and their constructive comments and suggestions.

\printbibliography

\appendix

\section{The quantum Threshold Decision problem}    \label{app:threshold-decision}
As mentioned at the end of \Cref{sec:threshold-search}, Aaronson~\cite{Aar20} showed that the \emph{decision} version of quantum Threshold Search can be done with $n = O(\log(m) \log(1/\delta)/\eps^2)$ copies, through the use of a theorem of Harrow, Lin, and Montanaro~\cite[Cor.~11]{HLM17}.
In \Cref{thm:quantum-or-technical} below, we give a new version of the Harrow--Lin--Montanaro theorem, with a mild qualitative improvement.
This improvement also lets us improve the quantum Threshold Decision copy complexity slightly, to $n = O(\log(m/\delta)/\eps^2)$ (see \Cref{cor:quantum-or}).

First, a lemma:
\begin{lemma}                                       \label{lem:CS}
    Let $X, Y \in \CC^{d \times d}$, with $X \geq 0$.  Then
    \[
        \E_{\rho}[XY] \leq \sqrt{\E_{\rho}[X]}\sqrt{\E_{\rho}[Y^\dagger X Y]}.
    \]
\end{lemma}
\begin{proof}
    This follows from the matrix form of Cauchy--Schwarz:
    \begin{align*}
         \tr(\rho X Y) = \tr\paren{\sqrt{\rho} \sqrt{X} \cdot \sqrt{X} Y \sqrt{\rho}}
         &\leq \sqrt{\tr\paren{\sqrt{X} \sqrt{\rho}  \sqrt{\rho} \sqrt{X}}}\sqrt{\tr\paren{\sqrt{\rho} Y^\dagger \sqrt{X}  \sqrt{X} Y \sqrt{\rho}}} \\
         &= \sqrt{\tr(\rho X)}\sqrt{\tr(\rho Y^\dagger X Y)}. \qedhere
    \end{align*}
\end{proof}

\begin{theorem}                                     \label{thm:quantum-or-technical}
    Let $0 \leq A_1, \dots, A_m \leq \Id$ be $d$-dimensional observables and define $\#A = A_1 + \cdots + A_m$.
    Let $\nu > 0$ and let $B$ be the orthogonal projector onto the span of eigenvectors of $\#A$ with eigenvalue at least~$\nu$.
    Then for any state $\rho \in \CC^{d \times d}$, writing $p_{\max} = \max_i \left\{\E_\rho[A_i]\right\}$, we have
    \[
        p_{\max} - 2\sqrt{\nu} \leq \E_\rho[B] \leq \E_{\rho}[\#A]/\nu.
    \]
\end{theorem}
\begin{remark}
One can read out a similar result in the work of Harrow, Lin, and Montanaro~\cite[Cor.~11]{HLM17}, except with a lower bound of $\E_\rho[B] \geq .632(p_{\max} - \nu)^2$.
Note that unlike our bound, their lower bound is never close to~$1$, even when $p_{\max}$ is very close to~$1$.
It is this difference that leads to our slight improvement for the Threshold Decision problem.
We speculate that the lower bound in our result can be sharpened further, to $(1-O(\sqrt{\nu}))p_{\max}$.
\end{remark}
\begin{proof}
    The upper bound in the theorem is just ``Markov's inequality''; it follows immediately from $\#A \geq \nu B$ (and $\rho \geq 0$).
    As for the lower bound, suppose $p_{\max} = \E_{\rho}[A_j] = 1 - \delta$.
    Using the notation $\ol{B} = \Id - B$, and defining $\beta = \E_{\rho}[\overline{B} A_j \overline{B}]$, we have
    \[
        \beta \leq \E_{\rho}[\ol{B} \cdot \#A \cdot \ol{B}] < \nu,
    \]
    since $A_j \leq \#A$  and $\ol{B} \cdot \#A \cdot \ol{B} < \nu \Id$ by definition.
    On the other hand, write $p = \E_{\rho}[\ol{B}]$, so our goal is to show $p < \delta + 2\sqrt{\nu}$.
    Then
     \begin{align*}
        p = \E_{\rho}[\ol{B}] = \E_{\rho}[A_j \cdot \ol{B}] + \E_{\rho}[\ol{A}_j \cdot \ol{B}]
        &\leq \sqrt{\E_{\rho}[A_j]}\sqrt{\E_{\rho}[\ol{B} A_j \ol{B}]} + \sqrt{\E_{\rho}[\ol{A}_j]}\sqrt{\E_{\rho}[\ol{B}\,\ol{A}_j \ol{B}]} \\
        &= \sqrt{1-\delta}\sqrt{\beta} + \sqrt{\delta}\sqrt{p - \beta},
    \end{align*}
    where the inequality is by  \Cref{lem:CS}, and the subsequent equality uses $p = \E_{\rho}[\ol{B} (A_j + \ol{A}_j) \ol{B}]$.
    The above deduction, together with $\beta < \nu$, yields an upper bound on~$p$.  Eschewing the tightest possible bound, we deduce from the above that
    \[
        p \leq \sqrt{\beta} + \sqrt{\delta}\sqrt{p} < \sqrt{\nu} + \frac{\delta + p}{2} \implies p \leq 2\sqrt{\nu} + \delta. \qedhere
    \]
\end{proof}

Given \Cref{thm:quantum-or-technical}, it's easy to obtain the following quantum Threshold Decision algorithm, similar to \cite[Lem.~14]{Aar20}:
\begin{corollary}                                       \label{cor:quantum-or}
    In the scenario of quantum \BajanSearch, suppose one only wishes to solve the \emph{decision problem}, meaning the algorithm has only two possible outputs:
    \begin{itemize}
        \item ``there exists~$j$ with $\E_\rho[A_j] > \theta_j - \eps$''; or else,
         \item ``\,$\E_\rho[A_i] \leq \theta_i$ for all~$i$''.
    \end{itemize}
    This can be solved using just $n = O(\log(m/\delta)/\eps^2)$ copies
    of~$\rho$ and probability of error at most $\delta$. The algorithm
    can be implemented by a projector applied to $\rho^{\otimes n}$.

    \noindent Furthermore, \Cref{cor:cor} holds.
\end{corollary}
\begin{proof}
    Writing $\varrho = \rho^{\otimes n}$, a standard Chernoff bound\rnote{make a proper reference to a generalized \Cref{lem:amplification}} implies there are quantum events $A'_1, \dots, A'_m$ satisfying
    \[
        \E_\rho[A_i] > \theta \implies \E_\varrho[A'_i] \geq 1 - \delta/2, \qquad
        \E_\rho[A_i] \leq \theta - \eps \implies \E_\varrho[A'_i] \leq \delta^3/(16m).
    \]
    We apply \Cref{thm:quantum-or-technical} to $A'_1, \dots, A'_m$ and $\varrho$, with $\nu = \delta^2/16$, obtaining the projector~$B$ with
    \[
        \max_i \{\E_{\varrho}[A'_i]\} -\delta/2 \leq \E_{\varrho}[B] \leq (16/\delta^2) \E_{\varrho}[\#A'].
    \]
    Now on one hand, if there exists $j$ with $\E_\rho[A_j] > \theta$, we conclude $\E_{\varrho}[B] \geq 1-\delta$.  On the other hand, if $\E_\rho[A_i] \leq \theta - \eps$ for all~$i$, then $\E_\varrho[\#A'] \leq m \cdot \delta^3/(16m)$ and hence $\E_{\varrho}[B] \leq \delta$.
    Thus the algorithm can simply measure~$B$ with respect to~$\varrho$, reporting  ``there exists~$j$ with $\E_\rho[A_j] > \theta - \eps$'' when~$B$ occurs, and ``\,$\E_\rho[A_j] \leq \theta$ for all~$i$'' when $\overline{B}$ occurs.
    This completes the main proof.

    The ``Furthermore'' proof of \Cref{cor:cor} is exactly the same, except we\rnote{Put in a reference to a suitably generalized \Cref{lem:amplification}.} let $A'_i$ be the quantum event that has
        \[
        \abs{\E_\rho[A_i] - \theta_i} > \eta_i + \eps \implies \E_\varrho[A'_i] \geq 1 - \delta/2, \qquad
        \abs{\E_\rho[A_i] - \theta_i} \leq \eta_i \implies \E_\varrho[A'_i] \leq \delta^3/(16m). \qedhere
    \]
\end{proof}

\end{document}